\documentclass[12pt]{article}

\usepackage{epsfig,amsmath,amssymb,amsfonts,amstext,amsthm}
\usepackage{latexsym,graphics,epsf,epsfig,color}
\usepackage{cite,url}

\newtheorem{theorem}{Theorem}
\newtheorem{prop}{Proposition}
\newtheorem{lemma}{Lemma}
\newtheorem{coro}[lemma]{Corollary}
\newtheorem{defi}{Definition}

\begin{document}

\title{Worst-Case Expected-Capacity Loss of\\ Slow-Fading Channels}

\author{Jae~Won~Yoo, Tie~Liu, Shlomo~Shamai~(Shitz), and Chao~Tian
\thanks{This paper was presented in part at the 2012 IEEE International Symposium on Information Theory, Cambridge, MA, USA, July 2012. This research was supported in part by the National Science Foundation under Grant CCF-08-45848 and by the Philipson Fund for Electrical Power, Technion Research Authority. J.~W.~Yoo and T.~Liu are with the Department of Electrical and Computer Engineering, Texas A\&M University, College Station, TX 77843, USA (email: \{yoojw78,tieliu\}@tamu.edu). S.~Shamai is with the Department of Electrical Engineering, Technion--Israel Institute of Technology, Haifa 32000, Israel (email: sshlomo@ee.technion.ac.il). C.~Tian is with the AT\&T Labs-Research, Florham Park, NJ 07932, USA (e-mail: tian@research.att.com).}
}

\date{\today}

\maketitle

\begin{abstract}
For delay-limited communication over block-fading channels, the difference between the ergodic capacity and the maximum achievable expected rate for coding over a finite number of coherent blocks represents a fundamental measure of the penalty incurred by the delay constraint. This paper introduces a notion of worst-case expected-capacity loss. Focusing on the slow-fading scenario (one-block delay), the worst-case additive and multiplicative expected-capacity losses are precisely characterized for the point-to-point fading channel. Extension to the problem of writing on fading paper is also considered, where both the ergodic capacity and the additive expected-capacity loss over one-block delay are characterized to within one bit per channel use.
\end{abstract}

\section{Introduction}
Consider the discrete-time baseband representation of the single-user flat-fading channel: 
\begin{equation}
Y[t]=\sqrt{G[t]}X[t]+Z[t]
\label{eq:Ch}
\end{equation}
where $\{X[t]\}$ are the channel inputs which are subject to a unit average power constraint, $\{G[t]\}$ are the power gains of the channel fading which we assume to be be \emph{unknown} to the transmitter but known at the receiver, $\{Z[t]\}$ are the additive white circularly symmetric complex Gaussian noise with zero means and unit variances, and $\{Y[t]\}$ are the channel outputs. As often done in the literature, we consider the so-called \emph{block-fading} model \cite[Ch.~5.4.5]{TV-B05} where $\{G[t]\}$ are assumed to be \emph{constant} within each coherent block and change \emph{independently} across different blocks according to a known distribution $F_G(\cdot)$. The coherent time of the channel is assumed to be large so that the additive noise $\{Z[t]\}$ can be ``averaged out" within each coherent block. Since both the power constraint and the noise variances are normalized to one, the power gain $G[t]$ also represents the instantaneous \emph{receive} signal-to-noise ratio of the channel.

The focus of this paper is on \emph{delay-limited} communication for which communication is only allowed to span (at most) a total of $L$ coherent blocks where $L$ is a \emph{finite} integer. In this setting, the Shannon capacity is a very pessimistic measure as it is dictated by the \emph{worst} realization of the power-gain process and hence equals zero when the realization of the power gain can be arbitrarily close to zero. An often-adopted measure in the literature is the \emph{expected capacity} \cite{Sha-ISIT97,SS-IT03,EG-ISIT98,EGL-IT10,VS-IT10}, which is defined as the maximum \emph{expected} reliably decoded rate where the expectation is over the distribution of the power-gain process.

The problem of characterizing the expected capacity is closely related to the problem of broadcasting over linear Gaussian channels \cite{Sha-ISIT97,SS-IT03,EG-ISIT98,EGL-IT10,VS-IT10}. The case with $L=1$ represents the most stringent delay requirement known as \emph{slow fading} \cite[Ch.~5.4.1]{TV-B05}. For slow-fading channels, the problem of characterizing the expected capacity is equivalent to the problem of characterizing the capacity region of a \emph{scalar} Gaussian broadcast channel, which is well understood based on the classical works of Cover \cite{Cov-IT72} and Bergmans \cite{Ber-IT74},  and then finding an optimal rate allocation based on the power-gain distribution. For $L>1$, the expected capacity can be improved by treating each realization of the power-gain process as a user in an \emph{$L$-parallel} Gaussian broadcast channel and coding the information bits across different sub-channels \cite{SS-IT03,WY-IT06,Ste-Tech06}. In the limit as $L \rightarrow \infty$, by the ergodicity of the power-gain process each ``typical" realization of the power-gain process can support a reliable rate of communication which is arbitrarily close to 
\begin{equation}
C_{erg}(F_G)=\mathbb{E}_G\left[\log(1+G)\right].
\label{eq:Ce}
\end{equation}
Thus, $C_{erg}(F_G)$ is both the Shannon capacity (appropriately known as the \emph{ergodic capacity} \cite[Ch.~5.4.5]{TV-B05}) and the expected capacity in the limit as $L \rightarrow \infty$.

Formally, let us denote by $C_{exp}(F_G,L)$ the expected capacity of the block-fading channel \eqref{eq:Ch} for which the power-gain distribution is $F_G(\cdot)$ and communication is allowed to span (at most) a total of $L$ coherent blocks. Then, as mentioned previously, the expected capacity $C_{exp}(F_G,L) \rightarrow C_{erg}(F_G)$ in the limit as $L \rightarrow \infty$. As such, the ``gap" between the ergodic capacity $C_{erg}(F_G)$ and the expected capacity $C_{exp}(F_G,L)$ represents a fundamental measure of the penalty incurred by imposing a delay constraint of $L$ coherent blocks. Such gaps, naturally, would depend on the underlying power-gain distribution. In this paper, we are interested in characterizing the \emph{worst-case} gaps over all possible power-gain distributions (including both the power-gain realizations and the probabilities for each realization) with a fixed number of different possible realizations of the power gain in each coherent block. 

More specifically, for the block-fading channel \eqref{eq:Ch} with the power-gain distribution $F_G(\cdot)$, let us define the \emph{additive} and the \emph{multiplicative} gap between the ergodic capacity and the expected capacity under the delay constraint of $L$ coherent blocks as
\begin{equation}
A(F_G,L) :=C_{erg}(F_G)-C_{exp}(F_G,L)
\end{equation}
and
\begin{equation}
M(F_G,L) :=\frac{C_{erg}(F_G)}{C_{exp}(F_G,L)}
\end{equation}
respectively. Focusing on the slow-fading scenario ($L=1$), we have the following precise characterization of the \emph{worst-case} additive and multiplicative gaps between the ergodic capacity and the expected capacity.

\begin{theorem}\label{thm:main}
\begin{equation}
\sup_{F_G} A(F_G,1) = \log K
\label{eq:AG}
\end{equation}
and
\begin{equation}
\sup_{F_G} M(F_G,1) = K
\label{eq:MG}
\end{equation}
where the supremes are over all power-gain distribution $F_G(\cdot)$ with $K$ different possible realizations of the power gain in each coherent block.
\end{theorem} 

The above results have both positive and negative engineering implications, which we summarize below.
\begin{itemize}
\item On the positive side, note that both the ergodic capacity $C_{erg}(F_G)$ and the expected capacity $C_{exp}(F_G,1)$ will grow \emph{unboundedly} in the limit as the realizations of the power gain all tend to infinity. The difference between them, however, will remain \emph{bounded} for any \emph{finite-state} fading channels (where $K$ is \emph{finite}). Similarly, both the ergodic capacity $C_{erg}(F_G)$ and the expected capacity $C_{exp}(F_G,1)$ will \emph{vanish} in the limit as the realizations of the power gain all tend to zero. However, the expected capacity $C_{exp}(F_G,1)$ (under the most stringent delay constraint of $L=1$ coherent block) can account, at least, for a \emph{non-vanishing} fraction of the ergodic capacity $C_{erg}(F_G)$. 
\item On the negative side, in the worst-case scenario both the additive gap $A(F_G,1)$ and the multiplicative gap $M(F_G,1)$ will grow \emph{unboundedly} in the limit as the number of different realizations of the power gain in each coherent block $K$ tends to infinity. Therefore, when $K$ is large, delay-limited communication may incur a large expected-rate loss relative to the ergodic scenario where there is no delay constraint on communication. For \emph{continuous-fading} channels where the sample space of $F_G(\cdot)$ is infinite and uncountable, it is possible that the expected-rate loss incurred by the delay constraint is \emph{unbounded}. 
\end{itemize}

One, however, should not be \emph{overly} pessimistic when attempt to interpret the worst-case gap results \eqref{eq:AG} and \eqref{eq:MG}. First, the above worst-case gap results are derived under the assumption that the transmitter does not know the realization of the channel fading at all. In practice, however, it is entirely possible that some information on the channel fading realization is made available to the transmitter (via finite-rate feedback, for example). This information can be potentially used to reduce the gap between the ergodic capacity and the expected capacity (over finite-block delay) \cite{SS-TWireless08a,SS-TWireless08b}. Second, for specific fading distributions the gap between the ergodic capacity and the expected capacity can be much smaller. For example, it is known \cite{SS-IT03} that for Rayleigh fading, the additive gap between the ergodic capacity and the expected capacity over one-block delay is only $1.649$ nats per channel use in the high receive signal-to-noise ratio limit, and the multiplicative gap is only $1.718$ in the low receive signal-to-noise ratio limit, even though in this case the power-gain distribution is continuous.

The rest of the paper is organized as follows. Next in Sec.~\ref{sec:Pf}, we provide a proof of the worst-case gap results \eqref{eq:AG} and \eqref{eq:MG} as stated in Theorem~\ref{thm:main}. Key to our proof is an explicit characterization of an optimal power allocation for characterizing the expected capacity $C_{exp}(F_G,1)$, obtained via the marginal utility functions introduced by Tse \cite{Tse-Tech99}. In Sec.~\ref{sec:FP}, we extend our setting from the point-to-point fading channel to the problem of writing on fading paper \cite{BB-IT08,BZ-ISIT06,ZKL-ISIT07}, and provide a characterization of the ergodic capacity and the additive expected-capacity loss over one-block delay to within one bit per channel use. Finally, in Sec.~\ref{sec:Con} we conclude the paper with some remarks.

\emph{Note}: In this paper, all logarithms are taken based on the natural number $e$.

\section{Proof of the Main Results} \label{sec:Pf}
\subsection{Optimal Power Allocation via Marginal Utility Functions}
To prove the worst-case gap results \eqref{eq:AG} and \eqref{eq:MG} as stated in Theorem~\ref{thm:main}, let us fix the transmit signal-to-noise ratio $1$ and the power-gain distribution $F_G(\cdot)$ with $K$ different possible realizations of the power gain in each coherent block. Let $\{g_1,\ldots,g_K\}$ be the collection of the possible realizations of the power gain, and let $p_k:=\mathrm{Pr}(G=g_k)>0$. Without loss of generality, let us assume that the possible realizations of the power gain are ordered as
\begin{equation}
g_1 > g_2  > \cdots > g_K \geq 0. \label{eq:Ord1}
\end{equation}
With the above notations, the expected capacity $C_{exp}(F_G,1)$ (under the delay constraint of $L=1$ coherent block) is given by \cite{SS-IT03}
\begin{equation}
\begin{array}{rl}
\max_{(\beta_1,\ldots,\beta_K)} & \sum_{k=1}^{K}F_k\log\left(\frac{1+\beta_kg_k}{1+\beta_{k-1}g_k}\right)\\
\mbox{subject to} & 0 = \beta_0 \leq \beta_1 \leq \beta_2 \leq \cdots \leq \beta_K \leq 1
\end{array}
\label{eq:P1}
\end{equation}
where 
\begin{equation}
F_k:=\sum_{j=1}^{k}p_j.
\end{equation}

Note that the  optimization program \eqref{eq:P1} with respect to the cumulative power vector $(\beta_1,\ldots,\beta_K)$ is \emph{not} convex. However, the program can be convexified via the following simple change of variable \cite{TSSD-IT08,Tse-Tech99}
\begin{equation}
r_k:=\log\left(\frac{1+\beta_kg_k}{1+\beta_{k-1}g_k}\right), \quad k=1,\ldots,K.
\end{equation}
In the preliminary version of this work \cite{YLS-ISIT12}, this venue was further pursued to obtain an \emph{implicit} characterization of the optimal power allocation via the standard Karush-Kuhn-Tucker conditions. Below we shall consider an alternative and  more direct approach which provides an \emph{explicit} characterization of an optimal power allocation via the \emph{marginal utility functions (MUFs)} introduced by Tse \cite{Tse-Tech99}.

Assume that $g_K>0$ (which implies that $g_k>0$ for all $k=1,\ldots,K$), and let $n_k:=1/g_k$ for $k=1,\ldots,K$. Given the assumed ordering \eqref{eq:Ord1} for the power-gain realizations $\{g_1,\ldots,g_K\}$, we have
\begin{equation}
0 < n_1 < \cdots < n_K.
\label{eq:Ord2}
\end{equation}
Following \cite{Tse-Tech99}, let us define the MUFs and the \emph{dominating} MUF as
\begin{equation}
u_k(z) := \frac{F_k}{n_k+z}, \quad k=1,\ldots,K
\end{equation}
and 
\begin{equation}
u^*(z) :=\max_{k=1,\ldots,K}u_k(z)
\end{equation}
respectively. Note that for any $k=1,\ldots,K$, $u_k(z)>0$ if and only if $z>-n_k$. Also, for any two distinct integers $k$ and $l$ such that $1 \leq k < l \leq K$ the MUFs $u_k(z)$ and $u_l(z)$ has a unique intersection at $z=z_{k,l}$ where
\begin{equation}
\frac{F_k}{n_k+z_{k,l}}=\frac{F_l}{n_l+z_{k,l}} \quad \Longleftrightarrow \quad z_{k,l} = \frac{F_kn_l-F_ln_k}{F_l-F_k}.
\label{eq:Int}
\end{equation}
Note that $F_k<F_l$ and $n_k < n_l$, so we have $z_{k,l} > -n_k$. Furthermore, it is straightforward to verify that $u_k(z) > u_l(z)>0$ if and only if $-n_k<z < z_{k,l}$, and $u_l(z) > u_k(z) >0$ if and only if $z > z_{k,l}$ (see Fig.~\ref{fig:SCP} for an illustration). For the rest of the paper, the above property will be frequently referred to as the \emph{single crossing point} property of the MUFs. 

\begin{figure}[t!]
\centering
\includegraphics[width=0.9\linewidth,draft=false]{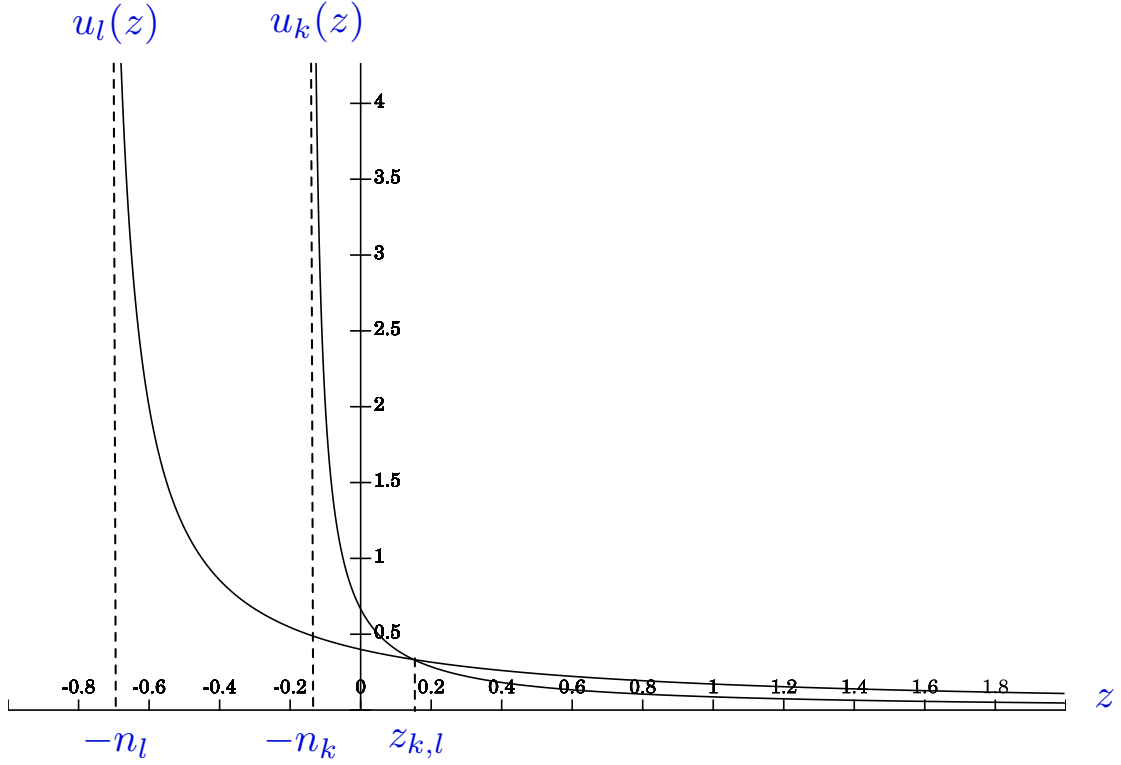}
\caption{The single crossing point property between the MUFs $u_k(z)$ and $u_l(z)$ for $k<l$.}
\label{fig:SCP}
\end{figure}

We emphasize here that the aforementioned single crossing point property relies on the fact that both sequences $\{n_k\}$ and $\{F_k\}$ increase monotonically with the subscript $k$. Since this particular ordering was not specifically considered in the MUFs defined in \cite[Eq.~(7)]{Tse-Tech99}, next, instead of building on the results from \cite{Tse-Tech99}, we shall borrow the concept of MUF and establish our results from first principles. Let us begin by defining a sequence of integers $\{\pi_1,\ldots,\pi_I\}$ recursively as follows. 

\begin{defi}
First, let $\pi_1=1$. Then, define
\begin{equation}
\pi_{i+1} := \max\left[\arg\min_{l=\pi_i+1,\ldots,K}z_{\pi_i,l}\right], \quad i=1,\ldots,I-1
\label{eq:RecDef}
\end{equation}
where $I$ is the total number of integers $\{\pi_i\}$ defined through the above recursive procedure. 
\end{defi}

Note that in the above definition, a ``$\max$" is used to break the ties for achieving the ``$\min$" inside the brackets, so there is no ambiguity in defining the integer sequence $\{\pi_1,\ldots,\pi_I\}$. Clearly, we have
\begin{equation}
1 = \pi_1 < \pi_2 < \cdots < \pi_I=K.
\end{equation}
Furthermore, we have the following properties for the sequence $\{z_{\pi_1,\pi_2},z_{\pi_2,\pi_3},\ldots,z_{\pi_{I-1},\pi_I}\}$, which are direct consequences of the recursive definition \eqref{eq:RecDef} and the single crossing point property of the MUFs.
\begin{lemma} \label{lemma:Mon} 
1) For any $i=1,\ldots,I-1$ and any $l =\pi_i+1,\ldots,K$, we have 
\begin{equation}
z_{\pi_{i},\pi_{i+1}} \leq z_{\pi_i,l}.
\label{eq:Mon1}
\end{equation}
2) For any $i=1,\ldots,I-2$, we have
\begin{equation}
z_{\pi_{i},\pi_{i+1}} \leq z_{\pi_{i+1},\pi_{i+2}}.
\label{eq:Mon2}
\end{equation}
3) For any $i=1,\ldots,I-1$ and any $l =1,\ldots,\pi_{i+1}-1$, we have 
\begin{equation}
z_{\pi_{i},\pi_{i+1}} \geq z_{l,\pi_{i+1}}.
\label{eq:Mon3}
\end{equation}
\end{lemma}

\begin{proof}
Property 1) follows directly from the recursive definition \eqref{eq:RecDef}.

To prove property 2), let us consider proof by contradiction. Assume that $z_{\pi_{i},\pi_{i+1}} > z_{\pi_{i+1},\pi_{i+2}}$ for some $i \in \{1,\ldots,I-2\}$. By property 1), we have $z_{\pi_{i},\pi_{i+2}} \geq z_{\pi_{i},\pi_{i+1}} > z_{\pi_{i+1},\pi_{i+2}}$. Following the single crossing point property, we have $0 < u_{\pi_{i+1}}(z_{\pi_i,\pi_{i+2}}) < u_{\pi_{i+2}}(z_{\pi_i,\pi_{i+2}})=u_{\pi_i}(z_{\pi_i,\pi_{i+2}})$. Using again the single crossing point property, we may conclude that $-n_{\pi_i}<z_{\pi_i,\pi_{i+2}}<z_{\pi_i,\pi_{i+1}}$. But this contradicts the fact that $z_{\pi_i,\pi_{i+2}} \geq z_{\pi_i,\pi_{i+1}}$ as mentioned previously. This proves that for any $i=1,\ldots,I-2$, we must have $z_{\pi_{i},\pi_{i+1}} \leq z_{\pi_{i+1},\pi_{i+2}}$.

To prove property 3), let us fix $i\in \{1,\ldots,I-1\}$. Note that the desired inequality \eqref{eq:Mon3} holds trivially with equality for $l=\pi_i$, so we only need to consider the cases where $l \in \{\pi_i+1,\ldots,\pi_{i+1}-1\}$ and $l \in \{1,\ldots,\pi_i-1\}$. 

For the case where $l \in \{\pi_i+1,\ldots,\pi_{i+1}-1\}$, by property 1) we have $-n_{\pi_i} < z_{\pi_i,\pi_{i+1}} \leq z_{\pi_i,l}$. Following the single crossing point property we have $0 < u_l(z_{\pi_{i},\pi_{i+1}}) \leq u_{\pi_i}(z_{\pi_{i},\pi_{i+1}})=u_{\pi_{i+1}}(z_{\pi_{i},\pi_{i+1}})$, which in turn implies that $z_{\pi_{i},\pi_{i+1}} \geq z_{l,\pi_{i+1}}$. 

For the case where $l \in \{1,\ldots,\pi_i-1\}$, let us assume, without loss of generality, that $l \in \{\pi_m,\ldots,\pi_{m+1}-1\}$ for some $m \in \{1,\ldots,i-1\}$. By the previous case we have $z_{\pi_{m},\pi_{m+1}} \geq z_{l,\pi_{m+1}} $ and hence 
\begin{equation}
0 < u_l(z) \leq u_{\pi_{m+1}}(z) \quad \forall z\geq z_{\pi_m,\pi_{m+1}}.
\label{eq:DMUF1}
\end{equation}
Also note that 
\begin{equation}
u_{\pi_{m+1}}(z) \leq u_{\pi_{m+2}}(z) \leq \cdots \leq u_{\pi_{i+1}}(z) \quad \forall z \geq \max_{m+1 \leq j \leq i}z_{\pi_j,\pi_{j+1}}.
\label{eq:DMUF2}
\end{equation}
By property 2) we have
\begin{equation}
\max_{m+1 \leq j \leq i}z_{\pi_j,\pi_{j+1}}=z_{\pi_{i},\pi_{i+1}} \geq z_{\pi_m,\pi_{m+1}}.
\label{eq:DMUF3}
\end{equation}
Combining \eqref{eq:DMUF1}--\eqref{eq:DMUF3} gives $0 < u_l(z_{\pi_{i},\pi_{i+1}}) \leq u_{\pi_{i+1}}(z_{\pi_{i},\pi_{i+1}})$, which in turn implies that $z_{\pi_{i},\pi_{i+1}} \geq z_{l,\pi_{i+1}}$.

Combing the above two cases completes the proof of property 3) and hence the entire lemma.
\end{proof}

\begin{figure}[t!]
\centering
\includegraphics[width=0.9\linewidth,draft=false]{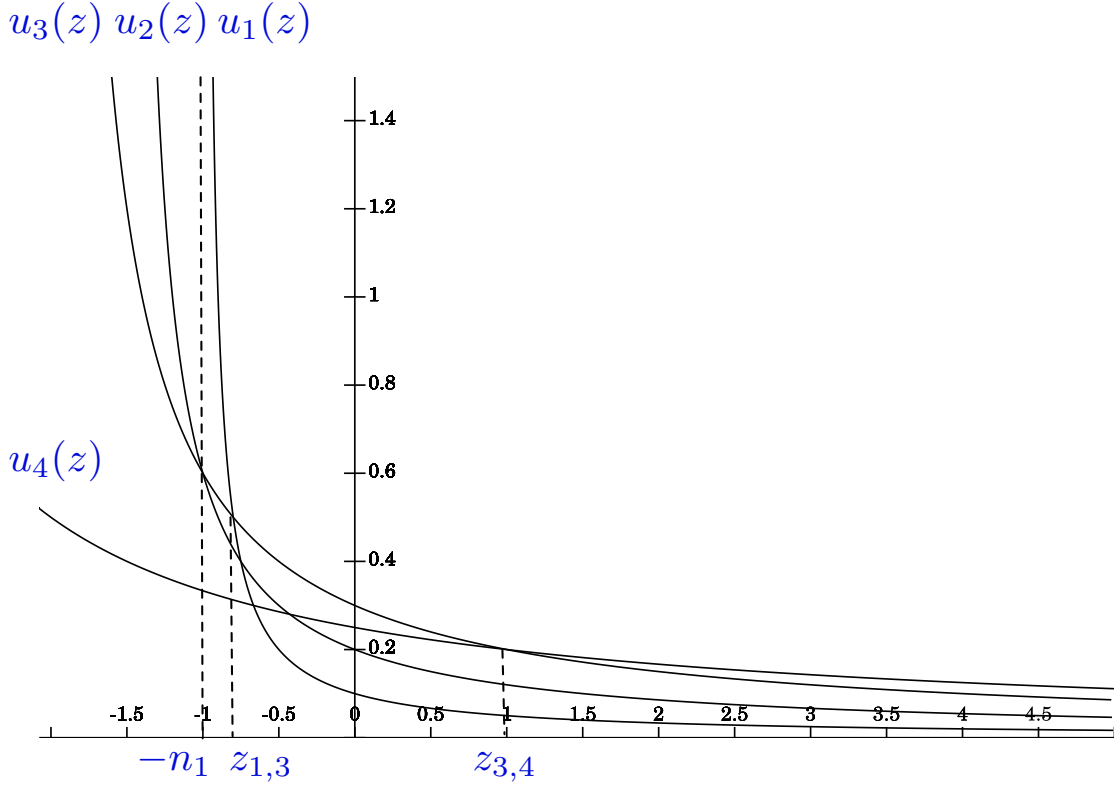}
\caption{An illustration of the dominating MUF. In this example, we have $K=4$ and $z_{1,3} < z_{1,2} < z_{1,4}$. Therefore, we have $I=3$, $\pi_1=1$, $\pi_2=3$, and $\pi_3=4$. The dominating MUF $u^*(z)=u_1(z)$ for $z \in (-n_1,z_{1,3})$, $u^*(z)=u_3(z)$ for $z \in (z_{1,3},z_{3,4})$, and $u^*(z)=u_4(z)$ for $z \in (z_{3,4},\infty)$.}
\label{fig:DMUF}
\end{figure}

The following proposition provides an explicit characterization of the dominating MUF (see Fig.~\ref{fig:DMUF} for an illustration).
\begin{prop}[Dominating marginal utility function]\label{prop:DMUF}
For any $i=1,\ldots,I$ and any $z \in (z_{\pi_{i-1},\pi_i},z_{\pi_{i},\pi_{i+1}})$, the dominating MUF
\begin{equation}
u^*(z)=u_{\pi_i}(z).
\end{equation}
where we define $z_{\pi_0,\pi_1}:=-n_1$ and $z_{\pi_I,\pi_{I+1}}:=\infty$ for notational convenience (even though $\pi_0$ and $\pi_{I+1}$ will not be explicitly defined).
\end{prop}

\begin{proof}
Fix $i \in \{1,\ldots,I\}$. Let us show that $u_{\pi_i}(z) \geq u_l(z)$ for any $z \in (z_{\pi_{i-1},\pi_i},z_{\pi_{i},\pi_{i+1}})$ by considering the cases $l>\pi_i$ and $l<\pi_i$ separately.

For $l>\pi_i$, by the single crossing point property we have $0 < u_{l}(z) \leq u_{\pi_i}(z)$ for any $-n_{\pi_i} < z \leq z_{\pi_i,l}$. By property 1) of Lemma~\ref{lemma:Mon}, for any $l>\pi_i$ we have $z_{\pi_i,\pi_{i+1}} \leq z_{\pi_i,l}$. Combined with the fact that $z_{\pi_{i-1},\pi_i} \geq -n_{\pi_i}$ (the equality holds only when $i=1$ by the definition of $z_{\pi_0,\pi_1}$ and the fact that $\pi_1=1$), we may conclude that for $l>\pi_i$, $u_{\pi_i}(z) \geq u_l(z)$ for any $z \in (z_{\pi_{i-1},\pi_i},z_{\pi_{i},\pi_{i+1}}]$.

For $l<\pi_i$, by property 3) of Lemma~\ref{lemma:Mon} we have $z_{\pi_{i-1},\pi_i} \geq z_{l,\pi_i}$ and hence $0 < u_l(z) \leq u_{\pi_i}(z)$ for any $z \geq z_{\pi_{i-1},\pi_i}$.

Combining the above two cases completes the proof of the proposition.
\end{proof}

Now, let $(\beta_1^*,\ldots,\beta_K^*)$ be an \emph{optimal} solution to the optimization program \eqref{eq:P1}. Then, the expected capacity $C_{exp}(F_G,1)$ can be bounded from above using the dominating MUF as follows:
\begin{eqnarray}
C_{exp}(F_G,1) & = & \sum_{k=1}^{K}F_k\log\left(\frac{n_k+\beta_k^*}{n_k+\beta_{k-1}^*}\right)\label{eq:OptPow1}\\
& = & \sum_{k=1}^{K}\int_{\beta_{k-1}^*}^{\beta_k^*}u_k(z)dz\\
& \leq & \sum_{k=1}^{K}\int_{\beta_{k-1}^*}^{\beta_k^*}u^*(z)dz\label{eq:LHY1}\\
& = & \int_{\beta_{0}^*}^{\beta_K^*}u^*(z)dz\\
& \leq & \int_{0}^{1}u^*(z)dz\label{eq:LHY2}
\end{eqnarray}
where \eqref{eq:LHY1} follows from the fact that for any $k=1,\ldots,K$ we have $\beta_{k-1}^* \leq \beta_k^*$ and $u_k(z) \leq u^*(z)$ for all $z$, and \eqref{eq:LHY2} follows from the fact that $\beta_0^*=0$, $\beta_K^* \leq 1$, and $u^*(z) >0 $ for all $z \geq 0$. The equalities hold if $(\beta_1^*,\ldots,\beta_K^*)$ satisfies 
\begin{equation}
u^*(z)=u_k(z) \quad \forall z \in (\beta_{k-1}^*,\beta_{k}^*)
\label{eq:Opt}
\end{equation} 
for any $k=1,\ldots,K$ and $\beta_K^*=1$.

Note that by property 3) of Lemma~\ref{lemma:Mon}, we have
\begin{equation}
-n_1=:z_{\pi_0,\pi_1} < z_{\pi_1,\pi_2} \leq \cdots \leq z_{\pi_{I-1},\pi_I} < z_{\pi_I,\pi_{I+1}}
:=\infty.
\end{equation}
To proceed, let us define two integers $s$ and $w$ as follows.

\begin{defi}\label{def:actFad}
Let $s$ be the \emph{largest} index $i \in \{1,\ldots,I\}$ such that $z_{\pi_{i-1},\pi_i}\leq 0$ and let $w$ be the \emph{largest} index $i \in \{1,\ldots,I\}$ such that $z_{\pi_{i-1},\pi_i} < 1$. 
\end{defi}

Clearly, we have $1 \leq s\leq w \leq I$. Furthermore if $s=w$, we have
\begin{equation}
\cdots \leq z_{\pi_{s-1},\pi_s} \leq 0 < 1 \leq z_{\pi_s,\pi_{s+1}} \leq \cdots
\label{eq:Ord100}
\end{equation}
and if $s<w$, we have 
\begin{equation}
\cdots \leq z_{\pi_{s-1},\pi_s} \leq 0 < z_{\pi_{s},\pi_{s+1}} \leq \cdots \leq  z_{\pi_{w-1},\pi_w} < 1 \leq z_{\pi_w,\pi_{w+1}} \leq \cdots
\label{eq:Ord200}
\end{equation}
Using the definition of $s$ and $w$, we have the following explicit characterization of an optimal power allocation.

\begin{prop}[An optimal power allocation]\label{prop:OptPow}
Assume that $g_K>0$. Then,  an optimal solution $(\beta_1^*,\ldots,\beta_K^*)$ to the optimization program \eqref{eq:P1} is given by
\begin{equation}
\beta_k^*=\left\{
\begin{array}{ll}
0, & \mbox{for} \quad 1 \leq k < \pi_s\\
z_{\pi_i,\pi_{i+1}}, & \mbox{for} \quad \pi_i \leq k < \pi_{i+1} \;\; \mbox{and} \;\; i=s,\ldots,w-1\\
1, & \mbox{for} \quad \pi_w \leq k \leq K.
\end{array}
\right.
\label{eq:OptPow}
\end{equation}
\end{prop}

\begin{figure}[t!]
\centering
\includegraphics[width=0.9\linewidth,draft=false]{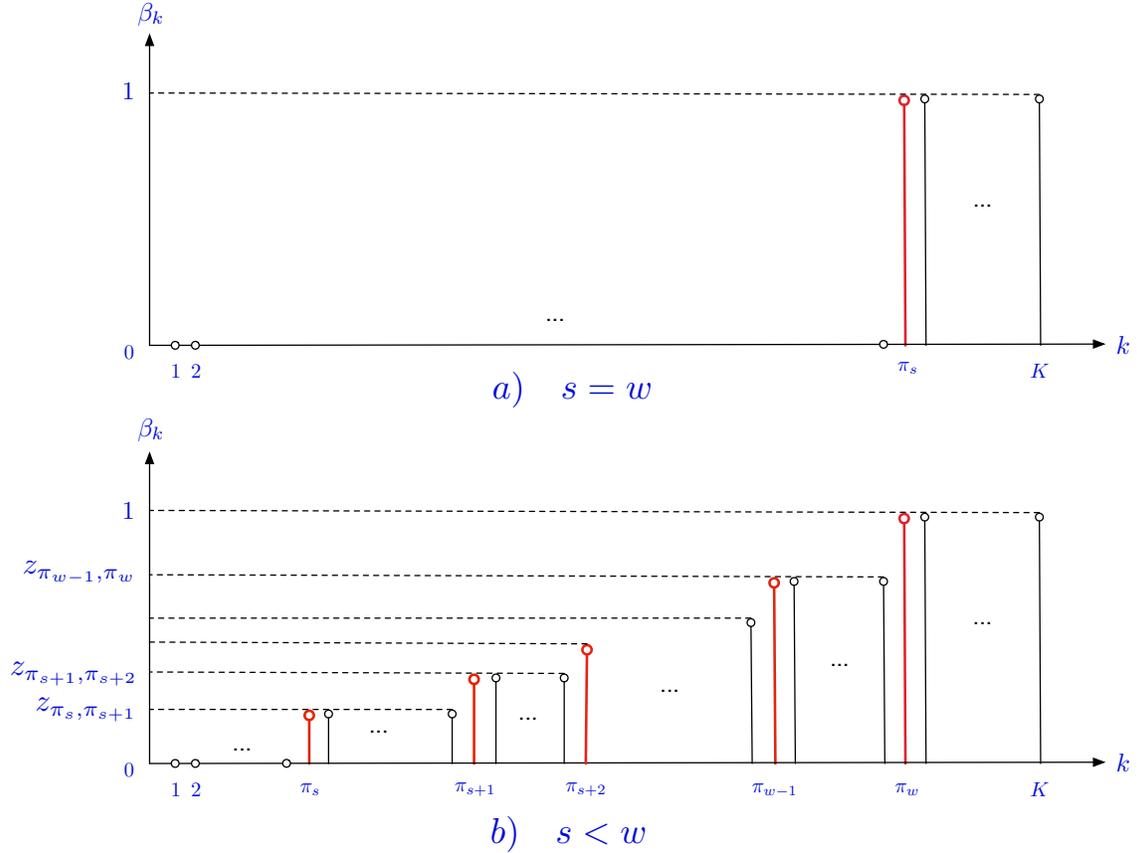}
\caption{An optimal power allocation obtained via the dominating MUF.}
\label{fig:OptPow}
\end{figure}

\begin{proof}
Note that we always have $\beta_K^*=1$. Therefore, in light of the previous discussion, it is sufficient to show that the choice of $(\beta_1^*,\ldots,\beta_K^*)$ as given by \eqref{eq:OptPow} satisfies \eqref{eq:Opt} for any $k=1,\ldots,K$. Also note that for the choice of \eqref{eq:OptPow}, we only need to consider the cases where $k=\pi_i$ for $i=s,\ldots,w$. Otherwise, we have $\beta_{k-1}^*=\beta_k^*$ so the open interval $(\beta_{k-1}^*,\beta_{k}^*)$ is empty and hence there is nothing to prove. 

Let us first assume that $s=w$. In this case, we only need to consider $k=\pi_s$, for which $\beta_{k-1}^*=0$ and $\beta_k^*=1$. By Proposition~\ref{prop:DMUF}, $u^*(z)=u_{\pi_s}(z)$ for any $z \in (z_{\pi_{s-1},\pi_s},z_{\pi_s,\pi_{s+1}})$. By \eqref{eq:Ord100}, $z_{\pi_{s-1},\pi_s}\leq 0$ and $z_{\pi_{s},\pi_{s+1}}\geq 1$. We thus conclude that $u^*(z)=u_{\pi_s}(z)$ for any $z \in (0,1)$.

Next, let us assume that $s<w$. We shall consider the following three cases separately.

Case 1: $k=\pi_s$. In this case, $\beta_{k-1}^*=0$ and $\beta_k^*=z_{\pi_s,\pi_{s+1}}$. By Proposition~\ref{prop:DMUF}, $u^*(z)=u_{\pi_s}(z)$ for any $z \in (z_{\pi_{s-1},\pi_s},z_{\pi_s,\pi_{s+1}})$. By \eqref{eq:Ord200}, $z_{\pi_{s-1},\pi_s}\leq 0$. We thus conclude that $u^*(z)=u_{\pi_s}(z)$ for any $z \in (0,z_{\pi_s,\pi_{s+1}})$.

Case 2: $k=\pi_i$ for some $i\in\{s+1,\ldots,w-1\}$. In this case, $\beta_{k-1}^*=z_{\pi_{i-1},\pi_i}$ and $\beta_k^*=z_{\pi_i,\pi_{i+1}}$. By Proposition~\ref{prop:DMUF}, $u^*(z)=u_{\pi_i}(z)$ for any $z \in (z_{\pi_{i-1},\pi_i},z_{\pi_i,\pi_{i+1}})$.

Case 3: $k=\pi_w$. In this case, $\beta_{k-1}^*=z_{\pi_{w-1},\pi_w}$ and $\beta_k^*=1$. By Proposition~\ref{prop:DMUF}, $u^*(z)=u_{\pi_w}(z)$ for any $z \in (z_{\pi_{w-1},\pi_w},z_{\pi_w,\pi_{w+1}})$. By \eqref{eq:Ord200}, $z_{\pi_w,\pi_{w+1}} \geq 1$. We thus conclude that $u^*(z)=u_{\pi_w}(z)$ for any $z \in (z_{\pi_{w-1},\pi_w},1)$.

We have thus completed the proof of the proposition.
\end{proof}

Note from \eqref{eq:P1} that the power allocated to the fading state $g_k$ is given by $\beta_k-\beta_{k-1}$. Thus for the optimal power allocation given by \eqref{eq:OptPow}, the ``active" fading states $g_k$ that are assigned to nonzero power (i.e., $\beta^*_k>\beta^*_{k-1}$) are $\pi_s,\pi_{s+1},\ldots,\pi_w$, i.e., $g_{\pi_s}$ is the strongest active fading state, and $g_{\pi_w}$ is the weakest active fading state (see Fig.~\ref{fig:OptPow} for an illustration). This provides an operational meaning for the integer sequence $\{\pi_1,\pi_2,\ldots,\pi_I\}$ and the integers $s$ and $w$ defined earlier.

Building on Proposition~\ref{prop:OptPow}, we have the following characterization of the expected capacity $C_{exp}(F_G,1)$, which will play a key role in proving the desired worst-case gap results \eqref{eq:AG} and \eqref{eq:MG}. The proof mainly involves some straightforward calculations and hence is deferred to Appendix~\ref{App:PA} to enhance the flow of the paper.

\begin{prop}[Expected capacity over one-block delay]\label{prop:Ce}
Assume that $g_K>0$ and let
\begin{equation}
\Lambda_k := \left\{
\begin{array}{ll}
\frac{n_{\pi_w}+1}{n_{\pi_s}}\frac{F_{\pi_s}}{F_{\pi_w}}, & \mbox{for} \;\; 1 \leq k \leq \pi_s\\
\frac{n_{\pi_w}+1}{n_{\pi_m}-n_{\pi_{m-1}}}\frac{F_{\pi_m}-F_{\pi_{m-1}}}{F_{\pi_w}}, & \mbox{for} \;\; \pi_{m-1} < k \leq \pi_m \;\; \mbox{and} \quad m=s+1,\ldots,w\\
1 & \mbox{for} \;\; \pi_w < k \leq K.
\end{array}
\right.
\label{eq:Ce2}
\end{equation} 
Then, the expected capacity $C_{exp}(F_G,1)$ can be written as
\begin{align}
C&_{exp}(F_G,1)\nonumber\\
& = \sum_{k=1}^Kp_k\log\Lambda_k\label{eq:Ce1.5}\\
& = F_{\pi_s}\log\left(\frac{F_{\pi_s}}{n_{\pi_s}}\right)+
\sum_{m=s+1}^{w}\left(F_{\pi_m}-F_{\pi_{m-1}}\right)\log\left(\frac{F_{\pi_m}-F_{\pi_{m-1}}}{n_{\pi_m}-n_{\pi_{m-1}}}\right)+F_{\pi_w}\log\left(\frac{n_{\pi_w}+1}{F_{\pi_w}}\right).\label{eq:Ce1.6}
\end{align}
\end{prop}

\subsection{Two Asymptotic Regimes}
Before we formally prove the worst-case gap results \eqref{eq:AG} and \eqref{eq:MG}, let us first take a look at the nature of the optimal power allocation \eqref{eq:OptPow} in two asymptotic regimes. As we shall see, these analyses provide some insights into why the worst-case additive and multiplicative gaps are $\log{K}$ and $K$, respectively.

Our first asymptotic analysis is in the \emph{high} receive signal-to-noise ratio regime and is motivated by the concept of \emph{generalized degree of freedom} \cite{ZT-IT03,ETW-IT08}. Consider
\begin{align}
g_k = \mathrm{SNR}^{r_k}, \quad k=1,\ldots,K
\end{align}
for some \begin{align}
r_1 > r_2 > \cdots > r_K >0\label{eq:Ord20}
\end{align}
where $\mathrm{SNR}$ can be made arbitrarily large. Fix $\{r_k\}$ and $\{p_k\}$. For sufficiently large $\mathrm{SNR}$, by \eqref{eq:Int} we have
\begin{align}
z_{k,l} &= \frac{F_k\mathrm{SNR}^{-r_l}-F_l\mathrm{SNR}^{-r_k}}{F_l-F_k} \approx \frac{F_k}{F_l-F_k}\mathrm{SNR}^{-r_l}
\label{eq:HighSNR}
\end{align}
for any $1 \leq k < l \leq K$. By the ordering \eqref{eq:Ord20}, we have for sufficiently large $\mathrm{SNR}$
\begin{align}
\frac{F_k}{F_l-F_k}\mathrm{SNR}^{-r_l} < \frac{F_k}{F_{l+1}-F_k}\mathrm{SNR}^{-r_{l+1}}
\end{align}
and hence 
\begin{align}
z_{k,l} < z_{k,l+1}
\end{align}
for any $1 \leq k < l \leq K-1$. By the definition \eqref{eq:RecDef}, we have $I=K$ and $\pi_i=i$ for all $i=1,\ldots,K$. Furthermore, by \eqref{eq:HighSNR} we have $0 < z_{k,l} <1$ for sufficiently large $\mathrm{SNR}$. Hence, by Definition~\ref{def:actFad} we have $s=1$ and $w=K$. We thus conclude that for sufficiently large $\mathrm{SNR}$ all fading states $g_k$, $k=1,\ldots,K$, are active fading states that are assigned to nonzero power. By \eqref{eq:Ce1.6} the expected capacity over one-block delay
\begin{align}
C_{exp}(F_G,1) &=  F_1\log\left(F_1\mathrm{SNR}^{r_1}\right)+
\sum_{m=2}^{K}(F_m-F_{m-1})\log\left(\frac{F_m-F_{m-1}}{\mathrm{SNR}^{-r_m}-\mathrm{SNR}^{-r_{m-1}}}\right)+\notag\\
& \hspace{13pt} F_K\log\left(\frac{\mathrm{SNR}^{-r_K}+1}{F_K}\right)\\
& \approx \left(\sum_{m=1}^{K}p_mr_m\right)\log\mathrm{SNR}+\sum_{m=1}^{K}p_m\log{p_m}
\end{align}
and by \eqref{eq:Ce} the ergodic capacity
\begin{align}
C_{erg}(F_G) &= \sum_{k=1}^{K}p_k\log\left(1+\mathrm{SNR}^{r_k}\right) \approx \left(\sum_{k=1}^{K}p_kr_k\right)\log\mathrm{SNR}
\end{align}
for sufficiently large $\mathrm{SNR}$. Thus, for sufficiently large $\mathrm{SNR}$ the additive gap
\begin{align}
A(F_G,1) \approx -\sum_{m=1}^{K}p_m\log{p_m} =: H(F_G) \leq \log{K}
\end{align}
for any $\{r_k\}$ and $\{p_k\}$, where $H(F_G)$ denotes the entropy of the power-gain distribution $F_G(\cdot)$, and the last inequality follows from the well-known fact that a uniform distribution maximizes the entropy subject to the cardinality constraint. This suggests that the \emph{worst-case} additive gap may be $\log{K}$.

Our second asymptotic analysis is in the \emph{low} receive signal-to-noise ratio regime and is motivated by the concept of \emph{channel capacity per unit cost} \cite{Ver-IT90}. Consider
\begin{align}
g_k = \alpha_k\mathrm{SNR}, \quad k=1,\ldots,K
\end{align}
for some \begin{align}
\alpha_1 > \alpha_2 > \cdots > \alpha_K >0\label{eq:Ord30}
\end{align}
where $\mathrm{SNR}$ can be made arbitrarily close to zero. Fix $\{\alpha_k\}$ and $\{p_k\}$. For sufficiently small $\mathrm{SNR}$, by \eqref{eq:Int} we have
\begin{align}
z_{k,l} &= \frac{F_k\alpha_l^{-1}-F_l\alpha_k^{-1}}{F_l-F_k}\frac{1}{\mathrm{SNR}}
\label{eq:LowSNR}
\end{align}
for any $1 \leq k < l \leq K$. Note that for sufficiently small $\mathrm{SNR}$ we have $z_{k,l}>1$ whenever it is positive. Thus, by Definition~\ref{def:actFad} we have $w=s$, i.e., the only active fading state is $g_{\pi_s}$, for sufficiently small $\mathrm{SNR}$. By \eqref{eq:Ce1.6} the expected capacity over one-block delay
\begin{align}
C_{exp}(F_G,1) &=  F_{\pi_s}\log\left(1+\alpha_{\pi_s}\mathrm{SNR}\right) \approx F_{\pi_s}\alpha_{\pi_s}\mathrm{SNR}
\end{align}
and by \eqref{eq:Ce} the ergodic capacity
\begin{align}
C_{erg}(F_G) &= \sum_{k=1}^{K}p_k\log\left(1+\alpha_k\mathrm{SNR}\right) \approx \left(\sum_{k=1}^{K}p_k\alpha_k\right)\mathrm{SNR}
\end{align}
for sufficiently small $\mathrm{SNR}$. By Lemma~\ref{lemma:Mon} and the fact that $w=s$ we have
\begin{align}
z_{k,\pi_s} \leq z_{\pi_{s-1},\pi_s} \leq 0 < 1 < z_{\pi_s,\pi_{s+1}} \leq z_{\pi_s,l}
\end{align}
for any $1 \leq k < \pi_s < l \leq K$, which implies that
\begin{align}
F_{\pi_s}\alpha_{\pi_s} \geq F_k\alpha_k, \quad \forall k=1,\ldots,K.
\end{align}
Thus, for sufficiently small $\mathrm{SNR}$ the multiplicative gap
\begin{align}
M(F_G,1) & \approx \frac{\sum_{k=1}^{K}p_k\alpha_k}{F_{\pi_s}\alpha_{\pi_s}} \leq \sum_{k=1}^{K}1=K
\end{align}
for any $\{\alpha_k\}$ and $\{p_k\}$, suggesting that the \emph{worst-case} multiplicative gap may be $K$.

\subsection{Additive Gap}
To prove the worst-case additive gap result \eqref{eq:AG}, we shall prove that $\sup_{F_G}A(F_G,1) \leq \log{K}$ and $\sup_{F_G}A(F_G,1) \geq \log{K}$ separately. 

\begin{prop}[Worst-case additive gap, converse part]\label{prop:AG}
For any power-gain distribution $F_G(\cdot)$ with $K$ different realizations of the power gain in each coherent block, we have 
\begin{equation}
A(F_G,1) \leq \log{K}. \label{eq:AG2}
\end{equation}
\end{prop}

\begin{proof}
Let us first prove the desired inequality \eqref{eq:AG2} for the case where $g_K>0$. In this case, by Proposition~\ref{prop:Ce} the additive gap $A(F_G,1)$ can be written as
\begin{eqnarray}
A(F_G,1) &=& \sum_{k=1}^Kp_k\log\left(\frac{n_k+1}{n_k}\right)-\sum_{k=1}^Kp_k\log\Lambda_k\\
&=& \sum_{k=1}^Kp_k\log\left(\frac{n_k+1}{n_k\Lambda_k}\right).
\label{eq:AG10}
\end{eqnarray}
We have the following lemma, whose proof is rather technical and hence is deferred to Appendix~\ref{App:AG}.
\begin{lemma} \label{lemma:AG}
For any $k=1,\ldots,K$, we have
\begin{equation}
\frac{n_k+1}{n_k\Lambda_k} \leq \frac{1}{p_k}.
\label{eq:AG11}
\end{equation}
\end{lemma}

Substituting \eqref{eq:AG11} into \eqref{eq:AG10}, we have
\begin{equation}
A(F_G,1) \leq \sum_{k=1}^Kp_k\log(1/p_k) = H(F_G) \leq \log{K}.
\end{equation}
This proves the desired inequality \eqref{eq:AG2} for the case where $g_K>0$.

For the case where $g_K=0$, let us consider a modified power-gain distribution $F'_G(\cdot)$ with probabilities $p'_k=p_k$ for all $k=1,\ldots,K$ and $g'_k=g_k$ for all $k=1,\ldots,K-1$. While we have $g_K=0$ for the original power-gain distribution $F_G(\cdot)$, we shall let $g'_K=\epsilon$ for some
\begin{equation}
0 < \epsilon < \min_{k=1,\ldots,K-1}\left[\frac{F_k}{(1-F_k)+n_k}\right].
\end{equation}
By \eqref{eq:Int}, this will ensure that 
\begin{equation}
z'_{k,K}=\frac{F_k/\epsilon-n_k}{1-F_k}>1, \quad \forall k=1,\ldots,K-1.
\end{equation} 
By the definition of $w'$, $z'_{\pi'_{w'-1},\pi'_{w'}}<1$ so we must have $\pi'_{w'} \neq K$ and hence $\pi'_{w'} < K$. By Proposition~\ref{prop:OptPow}, this implies that ${\beta'}^*_K={\beta'}^*_{K-1}$ so the fading state $g'_K$ are assigned to zero power for the given power allocation $({\beta'}^*_1,\ldots,{\beta'}^*_K)$. Hence, the given power allocation $({\beta'}^*_1,\ldots,{\beta'}^*_K)$ achieves the \emph{same} expected rate for both power-gain distributions $F_G(\cdot)$ and $F'_G(\cdot)$. Since $({\beta'}^*_1,\ldots,{\beta'}^*_K)$ is optimal for the power-gain distribution $F'_G(\cdot)$ but not necessarily so for $F_G(\cdot)$, we have
\begin{equation}
C_{exp}(F_G,1) \geq C_{exp}(F'_G,1).
\label{eq:JH200}
\end{equation}
On the other hand, improving the realizations of the power-gain can only improve the channel capacity\footnote{By the same argument, we also have $C_{exp}(F_G,1) \leq C_{exp}(F'_G,1)$ and hence $C_{exp}(F_G,1) = C_{exp}(F'_G,1)$, even though this direction of the inequality is not needed in the proof.}, so we have
\begin{align}
C_{erg}(F_G) & \leq C_{erg}(F'_G).\label{eq:JH300}
\end{align}
Combining \eqref{eq:JH200} and \eqref{eq:JH300} gives
\begin{align}
A(F_G,1) &=C_{erg}(F_G)-C_{exp}(F_G,1)\\
&\leq C_{erg}(F'_G)-C_{exp}(F'_G,1)\\
&=A(F_G',1)\\
& \leq \log{K}
\end{align}
where the last inequality follows from the previous case for which $g'_K=\epsilon>0$. This completes the proof for the case where $g_K=0$. 

Combing the above two cases completes the proof of Proposition~\ref{prop:AG}.
\end{proof}

\begin{prop}[Worst-case additive gap, forward part]\label{prop:AG2}
Fix $K$, and consider the power-gain distributions $F_G^{(d)}(\cdot)$ with
\begin{equation}
g_k=\sum_{j=1}^{K-k+1}d^j=\frac{d(d^{K-k+1}-1)}{d-1}
\end{equation}
for some $d > \max\{K-1,2\}$ and uniform probabilities $p_k = 1/K$ for all $k=1,\ldots,K$. For this particular parameter family of power-gain distributions, we have
\begin{equation}
\lim_{d \rightarrow \infty}A(F_G^{(d)},1) = \log{K}.
\end{equation}
\end{prop}

\begin{proof}
For the given power-gain distribution $F_G^{(d)}$, it is straightforward to calculate that for any $1 \leq k < l <K$
\begin{align}
\frac{n_k+z_{k,l}}{n_k+z_{k,l+1}} & =\frac{l-k+1}{l-k}\frac{d^{K-l}-1}{d^{K-l+1}-1}\frac{d^{l-k}-1}{d^{l-k+1}-1}
< \frac{l-k+1}{l-k}\frac{d^{l-k}-1}{d^{l-k+1}-1}.\label{eq:JH1}
\end{align}
where the last inequality follows from the fact that $d>1$. Since $l-k \geq 1$ and $d>2$, we have
\begin{align}
(l-k+1)(d^{l-k}-1)-(l-k)(d^{l-k+1}-1) &= \left[1-(l-k)(d-1)\right]d^{l-k}-1<0.\label{eq:JH2}
\end{align}
Substituting \eqref{eq:JH2} into \eqref{eq:JH1} gives 
\begin{align}
\frac{n_k+z_{k,l}}{n_k+z_{k,l+1}} & < 1
\end{align}
which immediately implies that $z_{k,l} < z_{k,l+1}$ for any $1 \leq k < l <K$. We thus have $I=K$ and $\pi_i=i$ for all $i=1,\ldots,K$. Since $d > \max\{K-1,2\}$, we have
\begin{equation}
z_{1,2} = \frac{(d-2)d^K+d}{(d-1)g_1g_2}>0
\end{equation}
and
\begin{equation}
z_{K-1,K} = \frac{(K-1)(d+d^2)-Kd}{d(d+d^2)} < \frac{K-1}{d} < 1
\end{equation}
so by definition we have $s=1$ and $w=K$. Thus, by the expression of $\Lambda_k$ from \eqref{eq:Ce2} we have
\begin{equation}
\Lambda_k = \left\{
\begin{array}{ll}
\frac{\left(\sum_{j=1}^{K}d^j\right)\left(1+d\right)}{K\cdot{d}}, & \quad k=1\\
\frac{\left(\sum_{j=1}^{K-k+1}d^j\right)\left(\sum_{j=1}^{K-k+2}d^j\right)\left(1+d\right)}{K\cdot{d}^{K-k+3}}, & \quad k=2,\ldots,K.\\
\end{array}
\right.
\end{equation}
It follows that
\begin{align}
\frac{n_1+1}{n_1\Lambda_1} &= K\cdot\frac{\left(1+\sum_{j=1}^{K}d^j\right)d}{\left(\sum_{j=1}^{K}d^j\right)\left(1+d\right)}\\
&= K\cdot \frac{d^{K+1}+O(d^K)}{d^{K+1}+O(d^K)}\\
& \rightarrow K
\label{eq:ASL1}
\end{align}
in the limit as $d \rightarrow \infty$ and
\begin{align}
\frac{n_k+1}{n_k\Lambda_k} &= K\cdot\frac{\left(1+\sum_{j=1}^{K-k+1}d^j\right)d^{K-k+3}}{\left(\sum_{j=1}^{K-k+1}d^j\right)\left(\sum_{j=1}^{K-k+2}d^j\right)\left(1+d\right)}\\
&= K\cdot \frac{d^{2(K-k)+4}+O(d^{2(K-k)+3})}{d^{2(K-k)+4}+O(d^{2(K-k)+3})}\\
& \rightarrow K
\label{eq:ASL2}
\end{align}
in the limit as $d \rightarrow \infty$ for any $k=2,\ldots,K$. A numerical example illustrating the convergence of \eqref{eq:ASL1} and \eqref{eq:ASL2} is provided in Fig.~\ref{fig:AG}. By \eqref{eq:AG10}, the additive gap
\begin{align}
A(F_G^{(d)},1) &= \sum_{k=1}^{K}p_k\log\left(\frac{n_k+1}{n_k\Lambda_k}\right)\\
& \rightarrow \sum_{k=1}^{K}\frac{1}{K}\log{K}\\
&= \log{K}
\end{align}
in the limit as $d \rightarrow \infty$. This completes the proof of Proposition~\ref{prop:AG2}.
\end{proof}

\begin{figure}[t!]
\centering
\includegraphics[width=0.9\linewidth,draft=false]{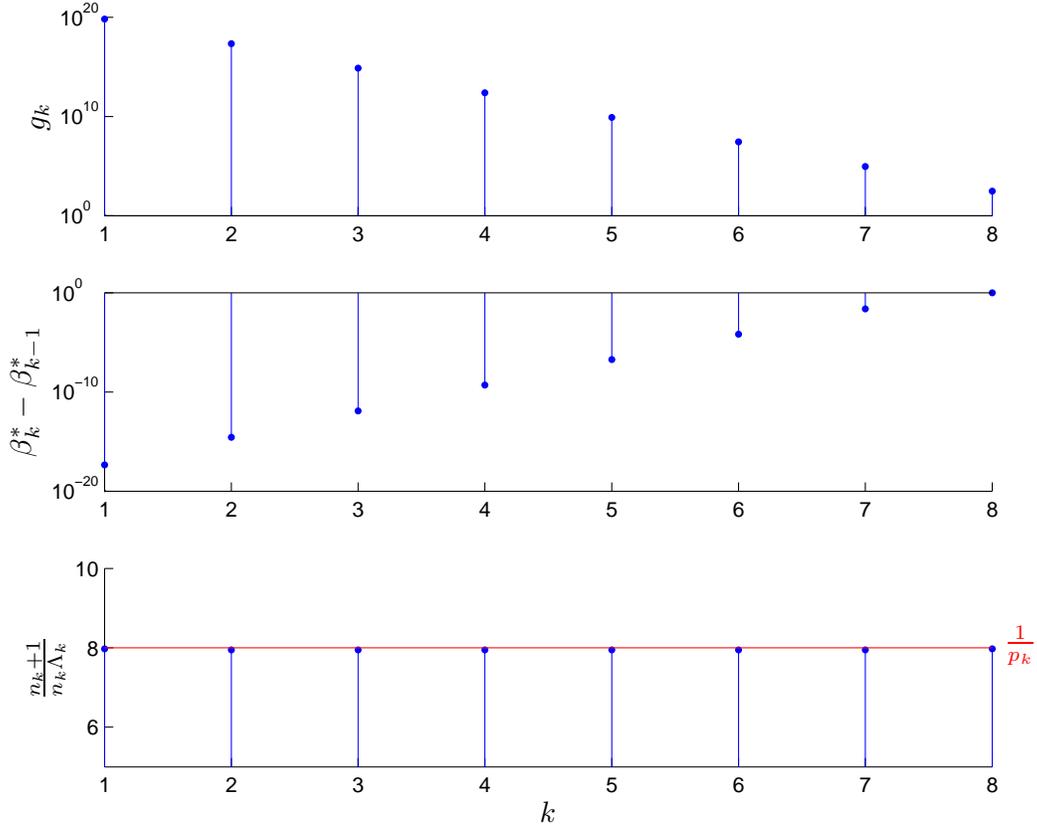}
\caption{A numerical example illustrating the convergence of \eqref{eq:ASL1} and \eqref{eq:ASL2}. In this example, $K=8$ and $d=300$.}
\label{fig:AG}
\end{figure}

Combining Propositions~\ref{prop:AG} and \ref{prop:AG2} completes the proof of the desired worst-case additive gap result \eqref{eq:AG}.

\subsection{Multiplicative Gap}
Similar to the additive case, to prove the worst-case multiplicative gap result \eqref{eq:MG} we shall prove that $\sup_{F_G}M(F_G,1) \leq K$ and $\sup_{F_G}M(F_G,1) \geq K$ separately. 

\begin{prop}[Worst-case multiplicative gap, converse part]\label{prop:MG}
For any power-gain distribution $F_G(\cdot)$ with $K$ different realizations of the power gain in each coherent block, we have 
\begin{equation}
M(F_G,1) \leq K. \label{eq:MG2}
\end{equation}
\end{prop}

\begin{proof}
Let us first prove the desired inequality \eqref{eq:MG2} for the case where $g_K>0$. By definition the multiplicative gap $M(F_G,1)$ can be written as
\begin{equation}
M(F_G,1) = \sum_{k=1}^K\frac{p_k\log\left(\frac{n_k+1}{n_k}\right)}{C_{exp}(F_G,1)}.
\label{eq:MG3}
\end{equation}
We have the following lemma, whose proof is deferred to Appendix~\ref{App:MG}.
\begin{lemma} \label{lemma:MG}
For any $k=1,\ldots,K$, we have
\begin{equation}
\frac{p_k\log\left(\frac{n_k+1}{n_k}\right)}{C_{exp}(F_G,1)} \leq 1.
\label{eq:MG4}
\end{equation}
\end{lemma}

Substituting \eqref{eq:MG4} into \eqref{eq:MG3}, we have
\begin{equation}
M(F_G,1) \leq \sum_{k=1}^K1 =K.
\end{equation}
This proves the desired inequality \eqref{eq:MG2} for the case where $g_K>0$.

For the case where $g_K=0$, we can use the same argument as for the additive case. More specifically, a modified power-gain distribution $F'_G(\cdot)$ can be found such that $g'_K>0$, $C_{exp}(F'_G,1) = C_{exp}(F_G,1)$, and $C_{erg}(F'_G) \geq C_{erg}(F_G)$. Thus, the multiplicative gap
\begin{align}
M(F_G,1) &=\frac{C_{erg}(F_G)}{C_{exp}(F_G,1)}\\
&\leq \frac{C_{erg}(F'_G)}{C_{exp}(F'_G,1)}\\
&=M(F_G',1)\\
& \leq K
\end{align}
where the last inequality follows from the previous case for which $g'_K>0$. This completes the proof for the case where $g_K=0$. 

Combing the above two cases completes the proof of Proposition~\ref{prop:MG}.
\end{proof}

\begin{prop}[Worst-case multiplicative gap, forward part]\label{prop:MG2}
Fix $K$, and consider the power-gain distributions $F_G^{(d)}(\cdot)$ with
\begin{equation}
n_k=\sum_{j=1}^{k}d^j
\end{equation}
for some $d >0$ and
\begin{equation}
p_k = \frac{d^k}{\sum_{j=1}^{K}d^j}
\end{equation} 
for all $k=1,\ldots,K$. For this particular parameter family of power-gain distributions, we have
\begin{equation}
\lim_{d \rightarrow \infty}M(F_G^{(d)},1) = K.
\end{equation}
\end{prop}

\begin{proof}
Note that for the given power-gain distribution $F_G^{(d)}$ pair,
\begin{equation}
F_k=\sum_{j=1}^{k}p_j=\frac{\sum_{j=1}^{k}d^j}{\sum_{j=1}^{K}d^j}
\end{equation}
so
\begin{equation}
z_{k,l}=\frac{F_kn_l-F_ln_k}{F_l-F_k}=0, \quad \forall 1 \leq k < l \leq K.
\end{equation}
We thus have $I=2$, $\pi_1=1$, $\pi_2=K$, and $s=w=2$. By the expression of $\Lambda_k$ from \eqref{eq:Ce2}, we have
\begin{equation}
\Lambda_k= \frac{n_K+1}{n_K}, \quad \forall k=1,\ldots,K.
\end{equation}
It follows that the expected capacity
\begin{equation}
C_{exp}(F_G^{(d)},1)=\sum_{k=1}^{K}p_k\log\Lambda_k=\log\left(\frac{n_K+1}{n_K}\right).
\end{equation}
We thus have
\begin{align}
\frac{p_k\log\left(\frac{n_k+1}{n_k}\right)}{C_{exp}(F_G^{(d)},1)} &=\frac{p_k\log\left(\frac{n_k+1}{n_k}\right)}{\log\left(\frac{n_K+1}{n_K}\right)}\\
& \geq \frac{p_kn_K}{n_k+1}\label{eq:MG5}\\
& = \frac{d^k}{\sum_{j=1}^kd^j+1}\\
& = \frac{d^k}{d^k+O(d^{k-1})}\\
& \rightarrow 1\label{eq:MG6}
\end{align}
in the limit as $d \rightarrow \infty$ for any $k=1,\ldots,K$, where \eqref{eq:MG5} follows from the well-known inequalities
\begin{equation}
\frac{x}{1+x} \leq \log(1+x) \leq x, \quad \forall x \geq 0, \label{eq:LogIn}
\end{equation}
so we have $\log\left(\frac{n_k+1}{n_k}\right) \geq \frac{1}{n_k+1}$ and $\log\left(\frac{n_K+1}{n_K}\right) \leq \frac{1}{n_K}$. On the other hand, by Lemma~\ref{lemma:MG}
\begin{align}
\frac{p_k\log\left(\frac{n_k+1}{n_k}\right)}{C_{exp}(F_G^{(d)},1)} &\leq 1\label{eq:MG8}
\end{align}
for any $k=1,\ldots,K$. Combining \eqref{eq:MG6} and \eqref{eq:MG8} proves that
\begin{equation}
\frac{p_k\log\left(\frac{n_k+1}{n_k}\right)}{C_{exp}(F_G^{(d)},1)} \rightarrow 1
\label{eq:MSL}
\end{equation}
in the limit as $d \rightarrow \infty$ for all $k=1,\ldots,K$. A numerical example illustrating the convergence of \eqref{eq:MSL} is illustrated in Fig.~\ref{fig:MG}. By \eqref{eq:MG3}, the multiplicative gap 
\begin{equation}
M(F_G^{(d)},1) = \sum_{k=1}^K\frac{p_k\log\left(\frac{n_k+1}{n_k}\right)}{C_{exp}(F_G^{(d)},1)} \rightarrow \sum_{k=1}^{K}1=K
\end{equation}
in the limit as $d \rightarrow \infty$. This completes the proof of Proposition~\ref{prop:MG2}.
\end{proof}

\begin{figure}[t!]
\centering
\includegraphics[width=0.9\linewidth,draft=false]{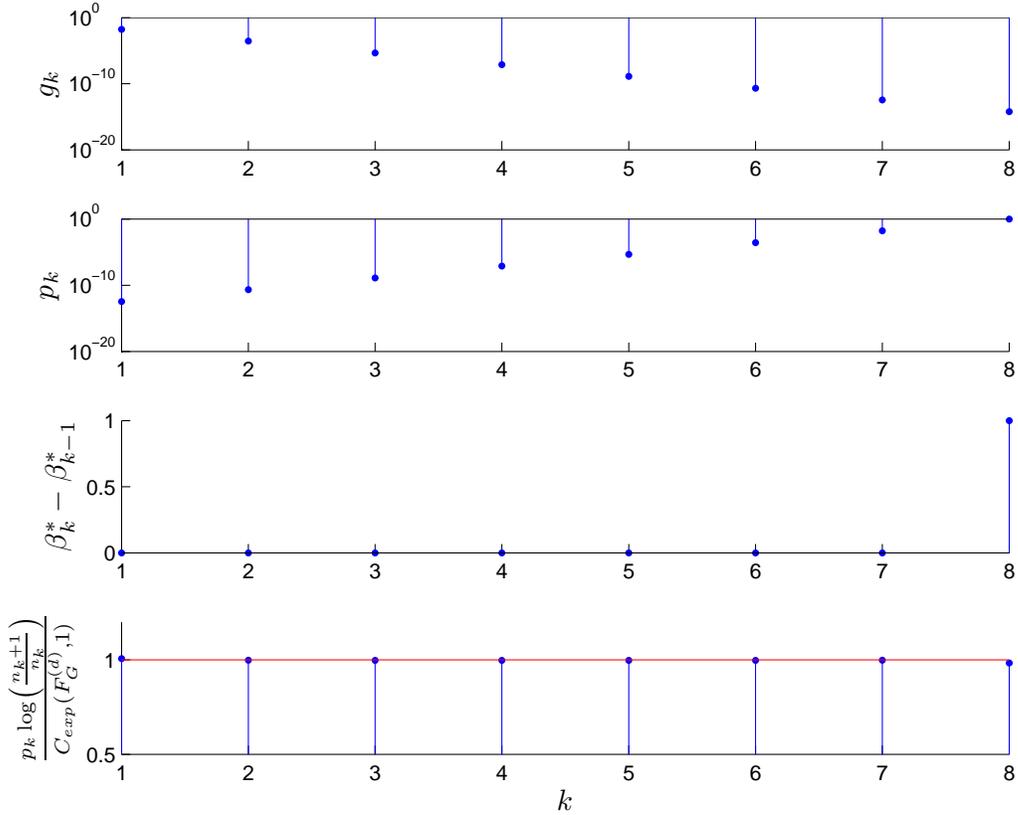}
\caption{A numerical example illustrating the convergence of \eqref{eq:MSL}. In this example, $K=8$ and $d=60$.}
\label{fig:MG}
\end{figure}

Combining Propositions~\ref{prop:MG} and \ref{prop:MG2} completes the proof of the desired worst-case multiplicative gap result \eqref{eq:MG}.

\section{Writing on Block-Fading Paper} \label{sec:FP}
Consider the problem of writing on fading paper \cite{BB-IT08,BZ-ISIT06,ZKL-ISIT07}:
\begin{equation}
Y[t]=\sqrt{G[t]}\left(X[t]+S[t]\right)+Z[t] \label{eq:Ch2}
\end{equation}
where $\{X[t]\}$ are the (complex) channel inputs which are subject to a unit average power constraint, $\{G[t]\}$ are the power gains of the channel fading which we assume to be \emph{unknown} to the transmitter but known at the receiver, $\{S[t]\}$ and $\{Z[t]\}$ are independent additive white circularly symmetric complex Gaussian interference and noise with zero means and variance $\mathrm{INR}$ and $1$ respectively, and $\{Y[t]\}$ are the channel outputs. The interference signal $\{S[t]\}$ are assumed to be non-causally known at the transmitter but not to the receiver. Note here that the instantaneous power gain $G[t]$ applies to both the channel input $X[t]$ and the known interference $S[t]$, so this model is particularly relevant to the problem of precoding for multiple-input multiple-output fading broadcast channels.

As for the point-to-point fading channel (\ref{eq:Ch}), we are interested in characterizing the worst-case expected-rate loss for the slow-fading scenario. However, unlike for the point-to-point fading channel (\ref{eq:Ch}), the ergodic capacity of the fading-paper channel (\ref{eq:Ch2}) is unknown. Below, we first characterize the ergodic capacity of the fading-paper model (\ref{eq:Ch2}) to within in one bit per channel use. As we will see, this will also lead to a characterization of the additive expected-capacity loss to within one bit per channel use for the slow-fading scenario.

\subsection{Ergodic Capacity to within One Bit}
Denote by $C_{erg}^{fp}(\mathrm{INR},F_G)$ the ergodic capacity of the fading-paper channel (\ref{eq:Ch2}) with the transmit interference-to-noise ratio $\mathrm{INR}$ and the power-gain distribution $F_G(\cdot)$. We have the following characterization of $C_{erg}^{fp}(\mathrm{INR},F_G)$ to within one bit.

\begin{theorem}\label{thm:fp1}
For any transmit interference-to-noise ratio $\mathrm{INR}$ and any power-gain distribution $F_G(\cdot)$, we have
\begin{equation}
C_{erg}(F_G)-\log2 \leq C_{erg}^{fp}(\mathrm{INR},F_G) \leq C_{erg}(F_G)\label{eq:fp1}
\end{equation}
where $C_{erg}(F_G)$ is the ergodic capacity of the point-to-point fading channel (\ref{eq:Ch}) of the same power-gain distribution as the fading-paper channel \eqref{eq:Ch2}.
\end{theorem}

\begin{proof}
To show that $C_{erg}^{fp}(\mathrm{INR},F_G) \leq C_{erg}(F_G)$, let us assume that the interference signal $\{S[t]\}$ are also known at the receiver. When the receiver knows both the power gain $\{G[t]\}$ and the interference signal $\{S[t]\}$, it can subtract $\{\sqrt{G[t]}S[t]\}$ from the received signal $\{Y[t]\}$. This will lead to an interference-free point-to-point fading channel (\ref{eq:Ch}), whose ergodic capacity is given by $C_{erg}(F_G)$. Since giving additional information to the receiver can only improve the ergodic capacity, we conclude that $C_{erg}^{fp}(\mathrm{INR},F_G) \leq C_{erg}(F_G)$.

To show that $C_{erg}^{fp}(\mathrm{INR},F_G) \geq C_{erg}(F_G)-\log{2}$, we shall show that
\begin{equation}
R=\mathbb{E}_G\left[(\log G)^+\right]
\label{eq:R}
\end{equation}
is an achievable ergodic rate for the fading-paper channel (\ref{eq:Ch2}), where $x^+:=\max(x,0)$. Since
\begin{equation}
(\log G)^+ \geq \log(1+G)-\log{2}
\end{equation}
for every possible realization of $G$, we will have
\begin{align}
C_{erg}^{fp}(\mathrm{INR},F_G) &\geq \mathbb{E}_G\left[(\log G)^+\right]\\
& \geq \mathbb{E}_G\left[\log(1+G)\right]-\log{2}\\
& = C_{erg}(F_G)-\log{2}.
\end{align}

To prove the achievability of the ergodic rate (\ref{eq:R}), we shall consider a communication scheme which is motivated by the following thought experiment. Note that with \emph{ideal} interleaving, the block-fading channel (\ref{eq:Ch2}) can be converted to a \emph{fast-fading} one \cite[Ch.~5.4.5]{TV-B05} for which the power gains $\{G[t]\}$ are independent across different time index $t$. Now that the channel is memoryless, by the well-known result of Gel'fand and Pinsker \cite{GP-PCIT80} the following ergodic rate is achievable:
\begin{equation}
R=\max_{(X,U)}\left[I(U;\sqrt{G}(X+S)+Z|G)-I(U;S)\right]
\end{equation}
where $U$ is an auxiliary variable which must be independent of $(G,Z)$. An optimal choice of the input-auxiliary variable pair $(X,U)$ is unknown \cite{BB-IT08,ZKL-ISIT07}. Motivated by the recent work \cite{ELGS-ITS11}, let us consider
\begin{equation}
U=X+S\label{eq:U}
\end{equation}
where $X$ is circularly symmetric complex Gaussian with zero mean and unit variance and is independent of $S$. For this choice of the input-auxiliary variable pair $(X,U)$, we have
\begin{align}
I&(U;\sqrt{G}(X+S)+Z|G)-I(U;S)\\
& = \mathbb{E}_G\left[\log(1+G(1+\mathrm{INR}))\right]-\log\left(1+\mathrm{INR}\right)\\
& \geq \mathbb{E}_G\left[\log(G(1+\mathrm{INR}))\right]-\log\left(1+\mathrm{INR}\right)\\
& = \mathbb{E}_G\left[\log G\right].\label{eq:JH1000}
\end{align}
This proves that
\begin{equation}
R=\left\{\mathbb{E}_G\left[\log G\right]\right\}^+
\label{eq:R2}
\end{equation}
is an achievable ergodic rate for the fading-paper channel (\ref{eq:Ch2}).

Note that even though the achievable ergodic rate (\ref{eq:R2}) is independent of the transmit interference-to-noise ratio $\mathrm{INR}$, it is \emph{not} always within one bit of the interference-free ergodic capacity $C_{erg}(F_G)$. This is because when $G<1$, we have $\log{G}<0$, i.e., the realizations of the power gain which are less than 1 contribute \emph{negatively} to the achievable rate (\ref{eq:R2}). By comparison, the realizations of the power gain \emph{never} contribute negatively (but possibly zero) to the achievable rate (\ref{eq:R}). Next, motivated by the secure multicast code construction proposed in \cite{KTW-IT08}, we shall consider a \emph{separate-binning} scheme which allows \emph{opportunistic} decoding at the receiver to boost the achievable ergodic rate from (\ref{eq:R2}) to (\ref{eq:R}).

Fix $\epsilon>0$ and let $(U,X)$ be chosen as in (\ref{eq:U}). Consider communicating a message $W \in \{1,\ldots,e^{LT_cR}\}$ over $L$ coherent blocks, each of a block length $T_c$ which we assume to be sufficiently large.

\begin{figure}[t!]
\centering
\includegraphics[width=0.7\linewidth,draft=false]{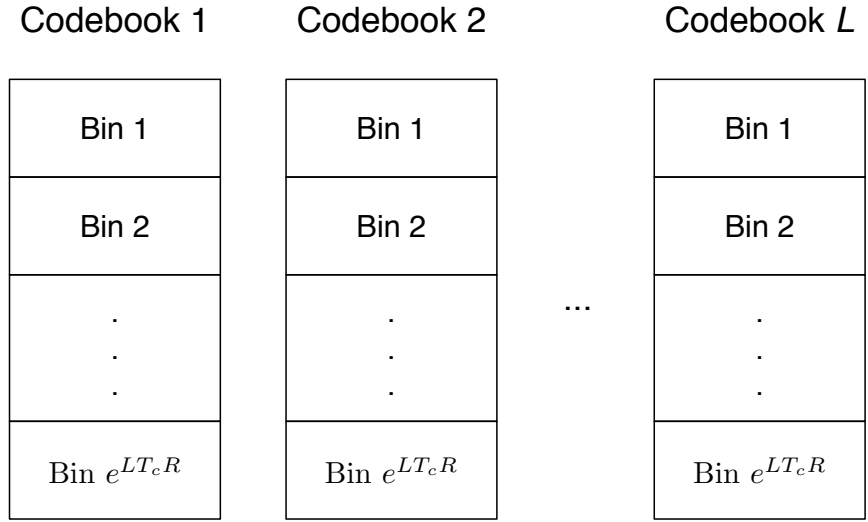}
\caption{The codebook structure for achieving the ergodic rate \eqref{eq:R}. Each codeword bin in the codebooks contains $e^{T_c(I(U;S)+\epsilon)}$ codewords.}
\label{fig}
\end{figure}

\emph{Codebook generation.} Randomly generate $L$ codebooks, each for one coherent block and consisting of $e^{T_c(LR+I(U;S)+\epsilon)}$ codewords of length $T_c$. The entries of the codewords are independently generated according to $P_U$. Randomly partition each codebook into $e^{LT_cR}$ bins, so each bin contains $e^{T_c(I(U;S)+\epsilon)}$ codewords. See Fig.~\ref{fig} for an illustration of the codebook structure.

\emph{Encoding.} Given the message $W$ and the interference signal $S^{LT_c}:=(S[1],\ldots,S[LT_c])$, the encoder looks into the $W$th bin in each codebook $l$ and tries to find a codeword that is jointly typical with $S_l^{T_c}$, where $S_l^{T_c}:=(S[(l-1)T_c+1],\ldots,S[lT_c])$ represents the segment of the interference signal $S^{LT_c}$ transmitted over the $l$th coherent block. By assumption, $T_c$ is sufficiently large so with high probability such a codeword can be found in each codebook \cite{Cos-IT83}. Denote by $U_l^{T_c}:=(U[(l-1)T_c+1],\ldots,U[lT_c])$ the codeword chosen from the $l$th codebook. The transmit signal $X_l^{T_c}:=(X[(l-1)T_c+1],\ldots,X[lT_c])$ over the $l$th coherent block is given by $X_l^{T_c}=U_l^{T_c}-S_l^{T_c}$.

\emph{Decoding.} Let $G_l$ be the realization of the power gain during the $l$th coherent block, and let
\begin{equation}
{\mathcal L}:=\{l: I(U;\sqrt{G_l}(X+S)+Z)-I(U;S) > 0\}.
\label{eq:L}
\end{equation}
Given the received signal $Y^{LT_c}:=(Y[1],\ldots,Y[LT_c])$, the decoder looks for a codeword bin which contains for \emph{each} coherent block $l \in \mathcal{L}$, a codeword that is jointly typical with the segment of $Y^{LT_c}(\mathcal{L})$ received over the $l$th coherent block. If only one such codeword bin can be found, the estimated message $\hat{W}$ is given by the index of the codeword bin. Otherwise, a decoding error is declared.

\emph{Performance analysis.} Note that averaged over the codeword selections and by the union bound, the probability that an incorrect bin index is declared by the decoder is no more than
\begin{equation}
\prod_{l \in \mathcal{L}}e^{T_c(I(U;S)+\epsilon)}\cdot e^{-T_c\left[I(U;\sqrt{G_l}(X+S)+Z)-\epsilon\right]}
= e^{-T_c\sum_{l \in \mathcal{L}}\left[I(U;\sqrt{G_l}(X+S)+Z)-I(U;S)-2\epsilon\right]}.
\end{equation}
Thus, by the union bound again, the probability of decoding error is no more than
\begin{equation}
e^{T_cLR}\cdot e^{-T_c\sum_{l \in \mathcal{L}}\left[I(U;\sqrt{G_l}(X+S)+Z)-I(U;S)-2\epsilon\right]}=
e^{-T_c\left\{\sum_{l \in \mathcal{L}}\left[I(U;\sqrt{G_l}(X+S)+Z)-I(U;S)-2\epsilon\right]-LR\right\}}.
\end{equation}
It follows that the transmit message $W$ can be reliably communicated (with exponentially decaying error probability for sufficiently large $T_c$) as long as
\begin{equation}
\sum_{l \in \mathcal{L}}\left[I(U;\sqrt{G_l}(X+S)+Z)-I(U;S)-2\epsilon\right]-LR > 0
\end{equation}
or equivalently
\begin{equation}
R < \frac{1}{L}\sum_{l \in {\mathcal L}}\left[I(U;\sqrt{G_l}(X+S)+Z)-I(U;S)-2\epsilon\right].
\end{equation}
Note that
\begin{align}
\frac{1}{L}\sum_{l \in {\mathcal L}}&\left[I(U;\sqrt{G_l}(X+S)+Z)-I(U;S)-2\epsilon\right]\\
& = \frac{1}{L}\sum_{l \in {\mathcal L}}\left[I(U;\sqrt{G_l}(X+S)+Z)-I(U;S)\right]-\frac{2|\mathcal{L}|}{L}\epsilon\\
& \geq \frac{1}{L}\sum_{l \in {\mathcal L}}\left[I(U;\sqrt{G_l}(X+S)+Z)-I(U;S)\right]-2\epsilon\label{eq:NFP1}\\
& = \frac{1}{L}\sum_{l=1}^{L}\left[I(U;\sqrt{G_l}(X+S)+Z)-I(U;S)\right]^+-2\epsilon\label{eq:NFP2}\\
&\geq \frac{1}{L}\sum_{l=1}^{L}(\log G_l)^+-2\epsilon\label{eq:JH2000}
\end{align}
where \eqref{eq:NFP1} follows from the fact that $|\mathcal{L}| \leq L$, \eqref{eq:NFP2} follows from the definition of $\mathcal{L}$ from \eqref{eq:L}, and \eqref{eq:JH2000} follows from \eqref{eq:JH1000}. Finally, by the weak law of large numbers,
\begin{equation}
\frac{1}{L}\sum_{l=1}^{L}(\log G_l)^+ \rightarrow \mathbb{E}_G\left[(\log G)^+\right]
\end{equation}
in probability in the limit as $L \rightarrow \infty$. We thus conclude that (\ref{eq:R}) is an achievable ergodic rate for the fading-paper channel (\ref{eq:Ch2}). 

We have thus completed the proof of Theorem~\ref{thm:fp1}.
\end{proof}

The following observations are now in place. First, the boost of the achievable rate from \eqref{eq:R2} to \eqref{eq:R} is mainly due to opportunistic decoding used by the receiver, which ensures that the realizations of the power gain which are less than $1$ do not contribute negatively to the achievable rate. Second, the separate-binning scheme takes advantage of the \emph{block-fading} nature and does not apply to the fast-fading scenario. Finally, the nature of the separate-binning scheme is such that the interference signal $S[t]$ within each coherent block only needs to be made available to the transmitter at the beginning of the block and not necessarily at the start of the entire communication.\footnote{This observation was brought to our attention by one of the referees during the review process.}

\subsection{Additive Expected-Capacity Loss to within One Bit}
Let $C_{exp}^{fp}(\mathrm{INR},F_G,L)$ be the expected capacity of the fading-paper channel (\ref{eq:Ch2}) under the delay constraint of $L$ coherent blocks, and let $A^{fp}(\mathrm{INR},F_G,L):=C_{erg}^{fp}(\mathrm{INR},F_G)-C_{exp}^{fp}(\mathrm{INR},F_G,L)$ be the additive gap between the ergodic capacity $C_{erg}^{fp}(\mathrm{INR},F_G)$ and the expected capacity $C_{exp}^{fp}(\mathrm{INR},F_G,L)$. We have the following results. 

\begin{theorem}\label{thm:fp2}
For any transmit interference-to-noise ratio $\mathrm{INR}$ and any power-gain distribution $F_G(\cdot)$, we have
\begin{equation}
A(F_G,1)-\log2 \leq A^{fp}(\mathrm{INR},F_G,1) \leq A(F_G,1).
\label{eq:fp2}
\end{equation}
\end{theorem} 

\begin{proof}
We claim that for any transmit interference-to-noise ratio $\mathrm{INR}$ and any power-gain distribution $F_G(\cdot)$, we have
\begin{equation}
C_{exp}^{fp}(\mathrm{INR},F_G,1)=C_{exp}(F_G,1).
\label{eq:fp3}
\end{equation}
Then, the desired inequalities in \eqref{eq:fp2} follow immediately from the above claim and Theorem~\ref{thm:fp1}.

To prove \eqref{eq:fp3}, let us consider the following $K$-user memoryless Gaussian broadcast channel:
\begin{equation}
Y_k = \sqrt{g_k}(X+S)+Z, \quad k=1,\ldots,K
\label{eq:GBC}
\end{equation}
where $X$ is the channel input which is subject an average power constraint, $S$ and $Z$ are independent additive white circularly symmetric complex Gaussian interference and noise, and $g_k$ and $Y_k$ are the power gain and the channel output of user $k$, respectively. The interference $S$ is assumed to be non-causally known at the transmitter but not to the receivers. Similar to the interference-free (scalar) Gaussian broadcast channel, the broadcast channel \eqref{eq:GBC} is also (stochastically) degraded. Furthermore, Steinberg \cite{Ste-IT05} showed that through \emph{successive} Costa precoding \cite{Cos-IT83} at the transmitter, the capacity region of the broadcast channel \eqref{eq:GBC} is the \emph{same} as that of the interference-free Gaussian broadcast channel. We may thus conclude that the expected capacity $C_{exp}^{fp}(\mathrm{INR},F_G,1)$ of the fading-paper channel \eqref{eq:Ch2} is the same as the expected capacity $C_{exp}(F_G,1)$ of the interference-free point-to-point fading channel \eqref{eq:Ch} of the same power-gain distribution $F_G(\cdot)$. This completes the proof of Theorem~\ref{thm:fp2}.
\end{proof}

Combining Theorems~\ref{thm:main} and \ref{thm:fp2} immediately leads to the following corollary.
\begin{coro}
\begin{equation}
\log(K/2) \leq \sup_{\mathrm{INR},F_G} A^{fp}(\mathrm{INR},F_G,1) \leq \log{K}.
\end{equation}
where the supreme is over all transmit interference-to-noise ratio $\mathrm{INR}$ and all power-gain distribution $F_G(\cdot)$ with $K$ different possible realizations of the power gain in each coherent block.
\end{coro}

\section{Concluding Remarks} \label{sec:Con}
For delay-limited communication over block-fading channels, the difference between the ergodic capacity and the maximum achievable expected rate for coding over a finite number of coherent blocks represents a fundamental measure of the penalty incurred by the delay constraint. This paper introduced a notion of worst-case expected-capacity loss. Focusing on the slow-fading scenario (one-block delay), it was shown that the worst-case additive expected-capacity loss is precisely $\log{K}$ nats per channel use and the worst-case multiplicative expected-capacity loss is precisely $K$, where $K$ is the total number of different possible realizations of the power gain in each coherent block. Extension to the problem of writing on fading paper was also considered, where both the ergodic capacity and the additive expected-capacity loss over one-block delay were characterized to within one bit per channel use.

Many research problems are open along the line of broadcasting over fading channels. Unlike for the case of one-block delay, the expected capacity of the point-to-point fading channel over \emph{multiple}-block delay is unknown except for the case with two-block delay and two different possible realizations of the power gain in each coherent block \cite{WY-IT06,Ste-Tech06}. The main difficulty there is that the capacity region of the \emph{parallel} Gaussian broadcast channel with a general message set configuration remains unknown. With multiple transmit antennas, the expected capacity of the point-to-point fading channel is unknown even for one-block delay \cite{SS-IT03}. Another interesting and challenging scenario is the \emph{mixed-delay} setting, where there are multiple messages of different delay requirement at the transmitter. Some preliminary results can be found in \cite{CSS-ISIT12}. With known interference at the transmitter, one may also consider the setting where the channel fading applies only to the known interference (the fading-dirt problem) \cite{KE-TWireless10} or, more generally, different channel fading applies to the input signal and the known interference separately.

\appendix
\section{Proof of Proposition~\ref{prop:Ce}}\label{App:PA}
Let us first rewrite the expression \eqref{eq:OptPow1} for the expected capacity $C_{exp}(F_G,1)$ as follows:
\begin{eqnarray}
C_{exp}(F_G,1) &=& \sum_{j=1}^K\left(\sum_{k=1}^jp_k\right)\log\left(\frac{n_j+\beta_j^*}{n_j+\beta_{j-1}^*}\right)\\
&=& \sum_{k=1}^Kp_k\left[\sum_{j=k}^K\log\left(\frac{n_j+\beta_j^*}{n_j+\beta_{j-1}^*}\right)\right]\\
&=& \sum_{k=1}^Kp_k\log\Lambda_k
\end{eqnarray}
where 
\begin{equation}
\Lambda_k = \prod_{j=k}^K\frac{n_j+\beta_j^*}{n_j+\beta_{j-1}^*}
\label{eq:Lambda}
\end{equation}
and $(\beta_1^*,\ldots,\beta_K^*)$ is given by \eqref{eq:OptPow}. 

To show that $\Lambda_k$ as given by \eqref{eq:Lambda} equals the right-hand side of \eqref{eq:Ce2}, let us first assume that $s=w$. For this case, by \eqref{eq:OptPow} we have $\beta_j^*=\beta_{j-1}^*$ for every $j \neq \pi_s$. Thus, substituting \eqref{eq:OptPow} into \eqref{eq:Lambda} gives
\begin{equation}
\Lambda_k=
\left\{
\begin{array}{ll}
\frac{n_{\pi_s}+1}{n_{\pi_s}}, & \mbox{for} \quad 1 \leq k \leq \pi_s\\
1, & \mbox{for} \quad \pi_s < k \leq K.
\end{array}
\right.
\end{equation}

Next, let us assume that $s<w$. We shall consider the following three cases separately. 

Case 1: $k \leq \pi_s$. For this case, substituting \eqref{eq:OptPow} into \eqref{eq:Lambda} gives
\begin{eqnarray}
\Lambda_k &=& \frac{n_{\pi_s}+z_{\pi_s,\pi_{s+1}}}{n_{\pi_s}}\left(\prod_{j=s+1}^{w-1}\frac{n_{\pi_j}+z_{\pi_j,\pi_{j+1}}}{n_{\pi_j}+z_{\pi_{j-1},\pi_j}}\right)\frac{n_{\pi_w}+1}{n_{\pi_w}+z_{\pi_{w-1},\pi_w}}\\
&=& \frac{n_{\pi_w}+1}{n_{\pi_s}}\prod_{j=s}^{w-1}\frac{n_{\pi_j}+z_{\pi_j,\pi_{j+1}}}{n_{\pi_{j+1}}+z_{\pi_j,\pi_{j+1}}}\\
&=& \frac{n_{\pi_w}+1}{n_{\pi_s}}\prod_{j=s}^{w-1}\frac{F_{\pi_j}}{F_{\pi_{j+1}}}\label{eq:AG3}\\
&=& \frac{n_{\pi_w}+1}{n_{\pi_s}}\frac{F_{\pi_s}}{F_{\pi_w}}
\end{eqnarray}
where \eqref{eq:AG3} follows from the fact that the MUFs $u_{\pi_j}(z)$ and  $u_{\pi_{j+1}}(z)$ intersect at $z=z_{\pi_j,\pi_{j+1}}$ so we have  
\begin{equation}
\frac{n_{\pi_j}+z_{\pi_j,\pi_{j+1}}}{F_{\pi_j}}=\frac{n_{\pi_{j+1}}+z_{\pi_j,\pi_{j+1}}}{F_{\pi_{j+1}}} \quad \Longleftrightarrow \quad \frac{n_{\pi_j}+z_{\pi_j,\pi_{j+1}}}{n_{\pi_{j+1}}+z_{\pi_j,\pi_{j+1}}}=\frac{F_{\pi_j}}{F_{\pi_{j+1}}}.
\label{eq:AG4}
\end{equation}

Case 2: $\pi_{m-1} < k \leq \pi_m$ for some $m \in \{s+1,\ldots,w\}$. For this case, substituting \eqref{eq:OptPow} into \eqref{eq:Lambda} gives
\begin{eqnarray}
\Lambda_k &=& \left(\prod_{j=m}^{w-1}\frac{n_{\pi_j}+z_{\pi_j,\pi_{j+1}}}{n_{\pi_j}+z_{\pi_{j-1},\pi_j}}\right)\frac{n_{\pi_w}+1}{n_{\pi_w}+z_{\pi_{w-1},\pi_w}}\\
&=& \frac{n_{\pi_w}+1}{n_{\pi_m}+z_{\pi_{m-1},\pi_m}}\prod_{j=m}^{w-1}\frac{n_{\pi_j}+z_{\pi_j,\pi_{j+1}}}{n_{\pi_{j+1}}+z_{\pi_j,\pi_{j+1}}}\\
&=& \frac{n_{\pi_w}+1}{n_{\pi_m}+z_{\pi_{m-1},\pi_m}}\prod_{j=m}^{e-1}\frac{F_{\pi_j}}{F_{\pi_{j+1}}}\label{eq:AG5}\\
&=& \frac{n_{\pi_w}+1}{n_{\pi_m}+z_{\pi_{m-1},\pi_m}}\frac{F_{\pi_m}}{F_{\pi_w}}\\
&=& \frac{n_{\pi_w}+1}{n_{\pi_m}-n_{\pi_{m-1}}}\frac{F_{\pi_m}-F_{\pi_{m-1}}}{F_{\pi_w}}\label{eq:AG6}
\end{eqnarray}
where \eqref{eq:AG5} follows from \eqref{eq:AG4}, and \eqref{eq:AG6} follows from the fact that the MUFs $u_{\pi_{m-1}}(z)$ and $u_{\pi_m}(z)$ intersect at $z=z_{\pi_{m-1},\pi_m}$ so by \eqref{eq:Int} we have  
\begin{equation}
z_{\pi_{m-1},\pi_m}=\frac{F_{\pi_{m-1}}n_{\pi_m}-F_{\pi_m}n_{\pi_{m-1}}}{F_{\pi_m}-F_{\pi_{m-1}}}
\quad \Longleftrightarrow \quad 
\frac{F_{\pi_m}}{n_{\pi_m}+z_{\pi_{m-1},\pi_m}}=\frac{F_{\pi_m}-F_{\pi_{m-1}}}{n_{\pi_m}-n_{\pi_{m-1}}}.
\end{equation}

Case 3: $k>\pi_w$. For this case, we have $\beta_j^*=\beta_{j-1}^*=1$ for any $j \geq k$. Hence, by \eqref{eq:OptPow} we have 
\begin{equation}
\Lambda_k=1.
\end{equation}

Finally, substituting \eqref{eq:Ce2} into \eqref{eq:Ce1.5} gives
\begin{align}
C_{exp}&(F_G,1)\nonumber\\
&= \sum_{k=1}^{\pi_s}p_k\log\Lambda_{\pi_s}+\sum_{m=s+1}^{w}\left(\sum_{k=\pi_{m-1}+1}^{\pi_m}p_k\right)\log\Lambda_{\pi_m}\\
&= F_{\pi_s}\log\Lambda_{\pi_s}+\sum_{m=s+1}^{w}\left(F_{\pi_{m}}-F_{\pi_{m-1}}\right)\log\Lambda_{\pi_m}\\
&= F_{\pi_s}\log\left(\frac{n_{\pi_w}+1}{n_{\pi_s}}\frac{F_{\pi_s}}{F_{\pi_w}}\right)+
\sum_{m=s+1}^{w}\left(F_{\pi_{m}}-F_{\pi_{m-1}}\right)\log\left(\frac{n_{\pi_w}+1}{n_{\pi_m}-n_{\pi_{m-1}}}\frac{F_{\pi_m}-F_{\pi_{m-1}}}{F_{\pi_w}}\right)\\
&= F_{\pi_s}\log\left(\frac{F_{\pi_s}}{n_{\pi_s}}\right)+
\sum_{m=s+1}^{w}\left(F_{\pi_m}-F_{\pi_{m-1}}\right)\log\left(\frac{F_{\pi_m}-F_{\pi_{m-1}}}{n_{\pi_m}-n_{\pi_{m-1}}}\right)+F_{\pi_w}\log\left(\frac{n_{\pi_w}+1}{F_{\pi_w}}\right).
\end{align}
This completes the proof of Proposition~\ref{prop:Ce}.

\section{Proof of Lemma~\ref{lemma:AG}}\label{App:AG}
Let us consider the following three cases separately.

Case 1: $k \leq \pi_s$. For such $k$, by property 3) of Lemma~\ref{lemma:Mon} and the definition of $s$ we have
\begin{equation}
\frac{F_kn_{\pi_{s}}-F_{\pi_{s}}n_k}{F_{\pi_{s}}-F_k}=z_{k,\pi_{s}} \leq z_{\pi_{s-1},\pi_{s}} \leq 0
\end{equation}
which implies that 
\begin{equation}
\frac{n_{\pi_s}}{F_{\pi_s}} \leq \frac{n_k}{F_k}. \label{eq:A3}
\end{equation}
By the expression of $\Lambda_k$ from \eqref{eq:Ce2}, for $k \leq \pi_s$ we have
\begin{eqnarray}
\frac{n_k+1}{n_k\Lambda_k} &=& \frac{n_k+1}{n_{\pi_w}+1}\frac{F_{\pi_w}n_{\pi_s}}{F_{\pi_s}n_k}\\
& \leq & \frac{n_k+1}{n_{\pi_w}+1}\frac{F_{\pi_w}}{F_k}\label{eq:A4}\\
& \leq & \frac{1}{p_k}\label{eq:A5}
\end{eqnarray}
where \eqref{eq:A4} follows from \eqref{eq:A3}, and \eqref{eq:A5} follows from the fact that $n_k+1 \leq n_{\pi_s}+1 \leq n_{\pi_w}+1$, $F_{\pi_w} \leq 1$, and $F_k \geq p_k$.

Case 2: $\pi_{m-1} < k \leq \pi_m$ for some $m \in \{s+1,\ldots,w\}$. For such $k$, by \eqref{eq:Ce2} we have
\begin{equation}
\frac{n_k+1}{n_k\Lambda_k} = \frac{n_k+1}{n_{\pi_w}+1}\frac{n_{\pi_m}-n_{\pi_{m-1}}}{F_{\pi_m}-F_{\pi_{m-1}}}\frac{F_{\pi_w}}{n_k}.
\label{eq:A6}
\end{equation}
By property 1) of Lemma~\ref{lemma:Mon} we have $z_{\pi_{m-1},\pi_m} \leq z_{\pi_{m-1},k}$ which implies that
\begin{equation}
\frac{n_{\pi_m}-n_{\pi_{m-1}}}{F_{\pi_m}-F_{\pi_{m-1}}}=\frac{n_{\pi_{m-1}}+z_{\pi_{m-1},\pi_m}}{F_{\pi_{m-1}}} \leq \frac{n_{\pi_{m-1}}+z_{\pi_{m-1},k}}{F_{\pi_{m-1}}}=\frac{n_{k}-n_{\pi_{m-1}}}{F_{k}-F_{\pi_{m-1}}}.
\label{eq:A7}
\end{equation}
Substituting \eqref{eq:A7} into \eqref{eq:A6} gives
\begin{equation}
\frac{n_k+1}{n_k\Lambda_k} \leq \frac{n_k+1}{n_{\pi_w}+1}\frac{n_{k}-n_{\pi_{m-1}}}{F_{k}-F_{\pi_{m-1}}}\frac{F_{\pi_w}}{n_k} \leq \frac{1}{p_k}\label{eq:A8}
\end{equation}
where the last inequality follows from the fact that $n_k+1 \leq n_{\pi_m}+1 \leq n_{\pi_w}+1$, $n_{k}-n_{\pi_{m-1}} \leq n_{k}$, $F_{\pi_w} \leq 1$, and $F_{k}-F_{\pi_{m-1}} \geq p_k$.

Case 3: $k > \pi_w$. For such $k$, by \eqref{eq:Ce2} we have $\Lambda_k=1$ and hence
\begin{equation}
\frac{n_k+1}{n_k\Lambda_k} = \frac{n_k+1}{n_k}.\label{eq:A10}
\end{equation}
By property 1) of Lemma~\ref{lemma:Mon} and the definition of $w$, we have $1 \leq z_{\pi_w,\pi_{w+1}} \leq z_{\pi_w,k}$, which implies that
\begin{equation}
n_k+1 \leq n_k+z_{\pi_w,k}=\frac{F_k(n_k-n_{\pi_w})}{F_k-F_{\pi_w}}.\label{eq:A11}
\end{equation}
Substituting \eqref{eq:A11} into \eqref{eq:A10} gives
\begin{equation}
\frac{n_k+1}{n_k\Lambda_k} \leq \frac{F_k}{F_k-F_{\pi_w}}\frac{n_k-n_{\pi_w}}{n_k} \leq \frac{1}{p_k}
\end{equation}
where the last inequality follows from the fact that $n_{k}-n_{\pi_w} \leq n_{k}$, $F_{k} \leq 1$, and $F_{k}-F_{\pi_w} \geq p_k$.

Combining the above three cases completes the proof of Lemma~\ref{lemma:AG}.

\section{Proof of Lemma~\ref{lemma:MG}}\label{App:MG}
Let us begin by establishing a simple lower bound on the expected capacity $C_{exp}(F_G,1)$. Applying the long-sum inequality
\begin{equation}
\sum_{i}a_i\log\frac{a_i}{b_i} \geq \left(\sum_{i}a_i\right)\log\frac{\sum_{i}a_i}{\sum_{i}b_i}\label{eq:logsum}
\end{equation}
we have 
\begin{equation}
F_{\pi_s}\log\left(\frac{F_{\pi_s}}{n_{\pi_s}}\right)+
\sum_{m=s+1}^{w}\left(F_{\pi_m}-F_{\pi_{m-1}}\right)\log\left(\frac{F_{\pi_m}-F_{\pi_{m-1}}}{n_{\pi_m}-n_{\pi_{m-1}}}\right) \geq F_{\pi_w}\log\left(\frac{F_{\pi_w}}{n_{\pi_w}}\right). \label{eq:B2}
\end{equation}
Substituting \eqref{eq:B2} into the expression of $C_{exp}(F_G,1)$ from \eqref{eq:Ce1.6}, we have
\begin{eqnarray}
C_{exp}(F_G,1) & \geq & F_{\pi_w}\log\left(\frac{F_{\pi_w}}{n_{\pi_w}}\right)+F_{\pi_w}\log\left(\frac{n_{\pi_w}+1}{F_{\pi_w}}\right)\\
& = & F_{\pi_w}\log\left(\frac{n_{\pi_w}+1}{n_{\pi_w}}\right).
\label{eq:B0}
\end{eqnarray}

Next we shall prove the desired inequality \eqref{eq:MG4} by considering the following four cases separately.

Case 1: $k>\pi_w$. For such $k$, by property 1) of Lemma~\ref{lemma:Mon} and the definition of $w$ we have $z_{\pi_w,k} \geq z_{\pi_w,\pi_{w+1}} \geq 1$ and hence
\begin{equation}
\frac{n_{\pi_w}+1}{F_{\pi_w}} \leq \frac{n_{\pi_w}+z_{\pi_w,k}}{F_{\pi_w}}=\frac{n_k-n_{\pi_w}}{F_k-F_{\pi_w}}.\label{eq:B1}
\end{equation}
Thus
\begin{eqnarray}
\frac{p_k\log\left(\frac{n_k+1}{n_k}\right)}{C_{exp}(F_G,1)} 
& \leq & \frac{p_k\log\left(\frac{n_k+1}{n_k}\right)}{F_{\pi_w}\log\left(\frac{n_{\pi_w}+1}{n_{\pi_w}}\right)}\label{eq:B3}\\
& \leq & \frac{p_k}{F_{\pi_w}}\frac{n_{\pi_w}+1}{n_k}\label{eq:B4}\\
& \leq & \frac{p_k}{n_k}\frac{n_k-n_{\pi_w}}{F_k-F_{\pi_w}}\label{eq:B5}\\
& \leq & 1\label{eq:B6}
\end{eqnarray}
where \eqref{eq:B3} follows from \eqref{eq:B0}, \eqref{eq:B4} is due to the well-know inequalities \eqref{eq:LogIn}
so $\log\left(\frac{n_k+1}{n_k}\right) \leq \frac{1}{n_k}$, and $\log\left(\frac{n_{\pi_w}+1}{n_{\pi_w}}\right) \geq \frac{1}{n_{\pi_w}+1}$, \eqref{eq:B5} follows from \eqref{eq:B1}, and \eqref{eq:B6} is due to the fact that $n_k-n_{\pi_w} \leq n_k$ and $F_k-F_{\pi_w}\geq p_k$.

Case 2: $k=\pi_w$. For such $k$, by \eqref{eq:B0} we have
\begin{equation}
\frac{\log\left(\frac{n_k+1}{n_k}\right)}{C_{exp}(F_G,1)} 
\leq \frac{\log\left(\frac{n_{\pi_w}+1}{n_{\pi_w}}\right)}{F_{\pi_w}\log\left(\frac{n_{\pi_w}+1}{n_{\pi_w}}\right)}= \frac{1}{F_{\pi_w}}\label{eq:B8}
\end{equation}
and hence 
\begin{equation}
\frac{p_k\log\left(\frac{n_k+1}{n_k}\right)}{C_{exp}(F_G,1)} \leq \frac{p_{\pi_w}}{F_{\pi_w}}\leq 1.\label{eq:B9}
\end{equation}

Case 3: $k=\pi_m$ for some $m \in \{s,\ldots,w-1\}$. For this case, we shall show that for any $m \in \{s,\ldots,w-1\}$
\begin{equation}
\frac{\log\left(\frac{n_{\pi_m}+1}{n_{\pi_m}}\right)}{C_{exp}(F_G,1)} \leq \frac{1}{F_{\pi_m}}
\label{eq:B10}
\end{equation}
and hence 
\begin{equation}
\frac{p_{\pi_m}\log\left(\frac{n_{\pi_m}+1}{n_{\pi_m}}\right)}{C_{exp}(F_G,1)} \leq \frac{p_{\pi_m}}{F_{\pi_m}} \leq 1.
\end{equation}

To prove \eqref{eq:B10}, let us define $g(z) :=N(z)/D(z)$ where
\begin{align}
N(z) &=\log\left(\frac{n_{\pi_m}+z}{n_{\pi_m}}\right) \label{eq:Nz}\\
D(z) &=F_{\pi_s}\log\left(\frac{F_{\pi_s}}{n_{\pi_s}}\right)+
\sum_{i=s+1}^{w}\left(F_{\pi_i}-F_{\pi_{i-1}}\right)\log\left(\frac{F_{\pi_i}-F_{\pi_{i-1}}}{n_{\pi_i}-n_{\pi_{i-1}}}\right)+F_{\pi_w}\log\left(\frac{n_{\pi_w}+z}{F_{\pi_w}}\right). \label{eq:Dz}
\end{align}
By Lemma~\ref{lemma:Mon} and the definition of $s$ and $w$, we have
\begin{equation}
0 < z_{\pi_s,\pi_{s+1}} \leq z_{\pi_m,\pi_{m+1}} \leq z_{\pi_m,\pi_w} \leq z_{\pi_{w-1},\pi_w} < 1.
\label{eq:B12}
\end{equation}
By the expression of $C_{exp}(F_G,1)$ from \eqref{eq:Ce1.6}, we have
\begin{equation}
\frac{\log\left(\frac{n_{\pi_m}+1}{n_{\pi_m}}\right)}{C_{exp}(F_G,1)} =g(1) \leq \sup_{z \geq z_{\pi_m,\pi_w}}g(z)\label{eq:B10.5}
\end{equation}
where the last inequality follows from the fact that $z_{\pi_m,\pi_w} < 1$ as mentioned in \eqref{eq:B12}. Next, we shall show that $g(z) \leq 1/F_{\pi_m}$ at the boundary points $z=z_{\pi_m,\pi_w}$ and $z=\infty$, and for any \emph{local maximum} $z^*>z_{\pi_m,\pi_w}$. We may then conclude that 
\begin{equation}
\sup_{z \geq z_{\pi_m,\pi_w}}g(z) \leq 1/F_{\pi_m}.\label{eq:B11}
\end{equation}

First, since $m<w$ we have
\begin{equation}
g(\infty)=1/F_{\pi_w}\leq 1/F_{\pi_m}.
\end{equation}

Next, to show that $g(z_{\pi_m,\pi_w}) \leq 1/F_{\pi_m}$, let us apply the log-sum inequality \eqref{eq:logsum} to obtain
\begin{equation}
F_{\pi_s}\log\left(\frac{F_{\pi_s}}{n_{\pi_s}}\right)+
\sum_{i=s+1}^{m}\left(F_{\pi_i}-F_{\pi_{i-1}}\right)\log\left(\frac{F_{\pi_i}-F_{\pi_{i-1}}}{n_{\pi_i}-n_{\pi_{i-1}}}\right) \geq F_{\pi_m}\log\left(\frac{F_{\pi_m}}{n_{\pi_m}}\right)\label{eq:B20}
\end{equation}
and
\begin{equation}
\sum_{i=m+1}^{w}\left(F_{\pi_i}-F_{\pi_{i-1}}\right)\log\left(\frac{F_{\pi_i}-F_{\pi_{i-1}}}{n_{\pi_i}-n_{\pi_{i-1}}}\right) \geq (F_{\pi_w}-F_{\pi_m})\log\left(\frac{F_{\pi_w}-F_{\pi_m}}{n_{\pi_w}-n_{\pi_m}}\right).\label{eq:B30}
\end{equation}
Substituting \eqref{eq:B20} and \eqref{eq:B30} into \eqref{eq:Dz} gives
\begin{align}
D(z_{\pi_m,\pi_w}) & \geq F_{\pi_m}\log\left(\frac{F_{\pi_m}}{n_{\pi_m}}\right)+(F_{\pi_w}-F_{\pi_m})\log\left(\frac{F_{\pi_w}-F_{\pi_m}}{n_{\pi_w}-n_{\pi_m}}\right)+F_{\pi_w}\log\left(\frac{n_{\pi_w}+z_{\pi_m,\pi_w}}{F_{\pi_w}}\right)\\
&= F_{\pi_m}\log\left(\frac{F_{\pi_m}}{n_{\pi_m}}\frac{n_{\pi_w}-n_{\pi_m}}{F_{\pi_w}-F_{\pi_m}}\right)+
F_{\pi_w}\log\left(\frac{F_{\pi_w}-F_{\pi_m}}{n_{\pi_w}-n_{\pi_m}}\frac{n_{\pi_w}+z_{\pi_m,\pi_w}}{F_{\pi_w}}\right)\\
&= F_{\pi_m}\log\left(\frac{n_{\pi_m}+z_{\pi_m,\pi_w}}{n_{\pi_m}}\right)\label{eq:B40}\\
&= F_{\pi_m}N(z_{\pi_m,\pi_w})\label{eq:B50}
\end{align}
where \eqref{eq:B40} follows from the fact that the MUFs $u_{\pi_m}(z)$ and $u_{\pi_w}(z)$ intersect at $z=z_{\pi_m,\pi_w}$ so we have
\begin{equation}
\frac{F_{\pi_m}}{n_{\pi_m}+z_{\pi_m,\pi_w}}=\frac{F_{\pi_w}}{n_{\pi_w}+z_{\pi_m,\pi_w}}=\frac{F_{\pi_w}-F_{\pi_m}}{n_{\pi_w}-n_{\pi_m}}.
\label{eq:B60}
\end{equation}
It follows immediately from \eqref{eq:B50} that 
\begin{equation}
g(z_{\pi_m,\pi_w})=N(z_{\pi_m,\pi_w})/D(z_{\pi_m,\pi_w})\leq1/F_{\pi_m}.
\end{equation}

Finally, to show that $g(z^*) \leq 1/F_{\pi_m}$ for any local maximum $z^* > z_{\pi_m,\pi_w}$, let us note that $g(z)$ is continuous and differentiable for all $z > z_{\pi_m,\pi_w}$ so $z^*$ must satisfy
\begin{equation}
\left.\frac{d}{dz}g(z)\right|_{z^*}=0
\end{equation}
or equivalently
\begin{equation}
\left.\frac{dN(z)}{dz}D(z)\right|_{z^*}=\left.\frac{dD(z)}{dz}N(z)\right|_{z^*}.
\end{equation}
We thus have 
\begin{align}
g(z^*) &=\frac{N(z^*)}{D(z^*)}=\left.\frac{dN(z)/dz}{dD(z)/dz}\right|_{z^*}\\
&=\frac{1}{F_{\pi_w}}\frac{n_{\pi_w}+z^*}{n_{\pi_m}+z^*}\\
&\leq \frac{1}{F_{\pi_w}}\frac{n_{\pi_w}+z_{\pi_m,\pi_w}}{n_{\pi_m}+z_{\pi_m,\pi_w}}\label{eq:B51}\\
&=\frac{1}{F_{\pi_m}}\label{eq:B52}
\end{align}
where \eqref{eq:B51} follows from the facts that $n_{\pi_w} > n_{\pi_m}$ so $\frac{n_{\pi_w}+z}{n_{\pi_m}+z}$ is a monotone decreasing function of $z$ for $z \geq 0$ and that $z^* \geq z_{\pi_m,\pi_w} > 0$, and \eqref{eq:B52} follows from \eqref{eq:B60}.

Substituting \eqref{eq:B11} into \eqref{eq:B10.5} completes the proof of the desired inequality \eqref{eq:B10} for Case 3.

Case 4: $k < \pi_w$ but $k \neq \pi_i$ for any $i=s,\ldots,w-1$. For such $k$, let $m$ be the \emph{smallest} integer from $\{s,\ldots,w\}$ such that $k<\pi_m$. Note that
\begin{eqnarray}
\frac{p_k\log\left(\frac{n_k+1}{n_k}\right)}{C_{exp}(F_G,1)} 
&=& \frac{p_k\log\left(\frac{n_k+1}{n_k}\right)}{\log\left(\frac{n_{\pi_m}+1}{n_{\pi_m}}\right)}
\frac{\log\left(\frac{n_{\pi_m}+1}{n_{\pi_m}}\right)}{C_{exp}(F_G,1)}\\
& \leq & \frac{p_k\log\left(\frac{n_k+1}{n_k}\right)}{F_{\pi_m}\log\left(\frac{n_{\pi_m}+1}{n_{\pi_m}}\right)}\label{eq:B61}\\
& =& \frac{p_k}{F_{\pi_m}}f(1)\label{eq:B62}
\end{eqnarray}
where \eqref{eq:B61} follows from \eqref{eq:B8} for $m=w$ and \eqref{eq:B10} for $m=s,\ldots,w-1$, and
\begin{equation}
f(z):=\frac{\log\left(\frac{n_k+z}{n_k}\right)}{\log\left(\frac{n_{\pi_m}+z}{n_{\pi_m}}\right)}.
\end{equation} 
Since $n_{k} < n_{\pi_m}$, $f(z)$ is a monotone decreasing function for $z > 0$. By Lemma~\ref{lemma:Mon} and the definition of $w$, we have
\begin{equation}
z_{k,\pi_m} \leq z_{\pi_{m-1},\pi_m} \leq z_{\pi_{w-1},\pi_w} < 1.\label{eq:B70}
\end{equation}
We shall consider the following two sub-cases separately.

Sub-case 4.1: $z_{k,\pi_m} >0$. By the monotonicity of $f(z)$ and the fact that $1>z_{k,\pi_m}>0$ as mentioned in \eqref{eq:B70}, we have
\begin{equation}
f(1) \leq f(z_{k,\pi_m})= \frac{\log\left(\frac{n_k+z_{k,\pi_m}}{n_k}\right)}{\log\left(\frac{n_{\pi_m}+z_{k,\pi_m}}{n_{\pi_m}}\right)}\leq \frac{n_{\pi_m}+z_{k,\pi_m}}{n_k}\label{eq:B80}
\end{equation}
where the last inequality follows from the inequalities \eqref{eq:LogIn} so we have $\log\left(\frac{n_k+z_{k,\pi_m}}{n_k}\right) \leq \frac{z_{k,\pi_m}}{n_k}$ and $\log\left(\frac{n_{\pi_m}+z_{k,\pi_m}}{n_{\pi_m}}\right) \geq \frac{z_{k,\pi_m}}{n_{\pi_m}+z_{k,\pi_m}}$. By Lemma~\ref{lemma:Mon} and the fact that $k < \pi_m$, we have $z_{\pi_{m-1},\pi_m} \geq z_{k,\pi_{m}}>0$ and hence $m \geq s+1$. Therefore, $k \neq \pi_{m-1}$ and we must have $k > \pi_{m-1}$. Again, by Lemma~\ref{lemma:Mon} we have $z_{k,\pi_{m}} \leq z_{\pi_{m-1},\pi_{m}} \leq z_{\pi_{m-1},k}$ and hence
\begin{equation}
\frac{n_{\pi_m}+z_{k,\pi_{m}}}{F_{\pi_m}}=\frac{n_{k}+z_{k,\pi_{m}}}{F_k} \leq \frac{n_k+z_{\pi_{m-1},k}}{F_k}=\frac{n_k-n_{\pi_{m-1}}}{F_k-F_{\pi_{m-1}}}.\label{eq:B90}
\end{equation}
Substituting \eqref{eq:B90} into \eqref{eq:B80} gives
\begin{equation}
f(1) \leq \frac{F_{\pi_m}(n_k-n_{\pi_{m-1}})}{n_k(F_k-F_{\pi_{m-1}})} \leq \frac{F_{\pi_m}}{F_k-F_{\pi_{m-1}}}.\label{eq:B100}
\end{equation}
Further substituting \eqref{eq:B100} into \eqref{eq:B62} gives
\begin{equation}
\frac{p_k\log\left(\frac{n_k+1}{n_k}\right)}{C_{exp}(F_G,1)}  \leq \frac{p_k}{F_k-F_{\pi_{m-1}}} \leq 1.
\end{equation}

Sub-case 4.2: $z_{k,\pi_m} \leq 0$. In this case, $z_{k,\pi_m} =\frac{F_kn_{\pi_m}-F_{\pi_m}n_k}{F_{\pi_m}-F_k}\leq 0$ so we have $F_kn_{\pi_m}\leq F_{\pi_m}n_k$. By the monotonicity of $f(z)$ we have
\begin{equation}
f(1) \leq \lim_{z \downarrow 0}f(z)=n_{\pi_m}/n_k\leq F_{\pi_m}/F_k.\label{eq:B110}
\end{equation}
Substituting \eqref{eq:B110} into \eqref{eq:B62} gives
\begin{equation}
\frac{p_k\log\left(\frac{n_k+1}{n_k}\right)}{C_{exp}(F_G,1)}  \leq \frac{p_k}{F_k} \leq 1.
\end{equation}

Combining the above two sub-cases completes the proof for Case 4. We have thus completed the proof of Lemma~\ref{lemma:MG}.

\section*{Acknowledgement}
Tie Liu would like to thank Dr. Jihong Chen for discussions that have inspired some ideas of the paper and two anonymous referees whose comments have helped to improve the presentation of the paper.

\end{document}